\documentclass[12pt]{article}
\usepackage{inputenc}
\usepackage{textcomp}
\usepackage{tgpagella}
\usepackage[margin=1.0in,letterpaper]{geometry}
\usepackage{setspace}
\usepackage{graphicx}
\usepackage{amsmath}
\usepackage{amssymb}
\usepackage{amsfonts}
\usepackage{amstext}
\usepackage{amsthm}
\usepackage{comment}
\usepackage[authoryear]{natbib}
\usepackage{enumerate}
\usepackage{todonotes}
\usepackage{url}
\usepackage{hyperref}
\usepackage{subfig}
\usepackage{multirow}
\usepackage{booktabs}
\usepackage[linesnumbered]{algorithm2e}
\usepackage{bm}

\theoremstyle{plain}
\newtheorem{lem}{\protect\lemmaname}
\theoremstyle{plain}
\newtheorem{assumption}{\protect\assumptionname}
\theoremstyle{plain}
\newtheorem{thm}{\protect\theoremname}
\theoremstyle{remark}

\theoremstyle{definition}
\newtheorem{definition}{\protect\definitionname}

\makeatother

\providecommand{\assumptionname}{Assumption}
\providecommand{\lemmaname}{Lemma}
\providecommand{\remarkname}{Remark}
\providecommand{\theoremname}{Theorem}
\providecommand{\definitionname}{Definition}

\global\long\def\cvp{\overset{p}{\to}}%
\global\long\def\cvw{\Rightarrow}%

\global\long\def\cvpst{\overset{p*}{\to}}%
\global\long\def\cvwst{\Rightarrow^*}%

\global\long\def\ion{\frac{1}{n}}%
\usepackage{array}
\newcolumntype{L}[1]{>{\raggedright\let\newline\\\arraybackslash\hspace{0pt}}m{#1}}
\newcolumntype{C}[1]{>{\centering\let\newline\\\arraybackslash\hspace{0pt}}m{#1}}
\newcolumntype{R}[1]{>{\raggedleft\let\newline\\\arraybackslash\hspace{0pt}}m{#1}}

\begin{document}


\title{Inference for parameters identified by conditional moment restrictions using a generalized Bierens maximum statistic\thanks{This is a revised version of the ArXiv:2008.11140 paper entitled ``Powerful Inference'', which has been online since 25 August 2020. We are grateful to the editor, four anonymous referees, and Jesse Shapiro for helpful comments. This work was supported in part by the Cowles Foundation, the European Research Council (ERC-2014-CoG- 646917-ROMIA) and by the UK Economic and Social Research Council (ESRC) through research grant (ES/P008909/1) to the CeMMAP.
Part of this research was carried out when Seo was visiting Cowles Foundation in 2018/2019 and was supported by the Ministry of Education of the Republic of Korea and the National Research Foundation of Korea (NRF-2018S1A5A2A01033487) and the Seoul National University Research Grant in 2023.}}	
\author{
Xiaohong Chen\footnote{Yale University and Cowles Foundation for Research in Economics; xiaohong.chen@yale.edu} 
 \and
Sokbae Lee\footnote{Columbia University and Centre for Microdata Methods and Practice; sl3841@columbia.edu}
 \and
Myung Hwan Seo\footnote{Author Correspondance; Seoul National University and Institute of Economic Research; myunghseo@snu.ac.kr}
\and
Myunghyun Song\footnote{Columbia University
	; ms6347@columbia.edu} }
\date{}

\maketitle


\doublespacing

\begin{abstract}
Many economic panel and dynamic models, such as rational behavior and Euler equations, imply that the parameters of interest are identified by conditional moment restrictions. We introduce a novel inference method without any prior information about which conditioning instruments are weak or irrelevant. Building on Bierens (1990), we propose penalized maximum statistics and combine bootstrap inference with model selection. Our method optimizes asymptotic power by solving a data-dependent max-min problem for tuning parameter selection. Extensive Monte Carlo experiments, based on an empirical example, demonstrate the extent to which our inference procedure is superior to those available in the literature.

\medskip

\noindent \textbf{Keywords}: conditional moment restrictions, conditional instruments, hypothesis testing, penalization, multiplier bootstrap, max-min.
\end{abstract}


\section{Introduction}

Conditional moment restrictions of the following forms are ubiquitous in economics:
\begin{align}
 & \mathbb{E} \left[ g\left(X_{i},\theta\right)|W_{i} \right ] =0 \quad a.s.\quad
 \text{if and only if }\theta=\theta_{0}, \label{eq:CM model}
\end{align}
where $g(x, \theta)$ is a real-valued known function with a finite-dimensional parameter $\theta$, and $W_i\in \mathbb{R}^p$ is a $p$-dimensional vector of conditioning instruments so that the parameter of interest $\theta_0$ is strongly identified by \eqref{eq:CM model}. There is now a mature literature on optimal estimation and inference for $\theta_0$ in \eqref{eq:CM model} when the conditioning instruments $W_i$ are all relevant and of low-dimension. Most existing approaches are firstly to estimate $\theta_0$ efficiently and secondly to develop suitable test statistics based on the estimators \citep[see, e.g.,][among many others]{Chamberlain:1987,donald2003empirical,Ai:Chen:03,dominguez2004consistent,kitamura2004empirical}. However, when the dimension $p$ of conditioning instruments $W_i$ is \emph{boundedly large}\footnote{We will use expression ``boundedly large'' throughout the paper to emphasize that we focus on the setting where the dimension ($p$) of the conditioning instruments is relatively large but does not grow with sample size ($n$). We are grateful to the editor for suggesting this expression.} and when some of $W_{i}$ are weak or irrelevant, semiparametrically efficient estimators for $\theta_0$, though theoretically optimal according to asymptotic theory, could perform poorly in small samples. In this paper, we take a different path and conduct inference for $\theta_0$ in \eqref{eq:CM model} directly by skipping the first step efficient estimation of $\theta_0$. 


To construct a more informative confidence set for $\theta_0$, 
 we propose an $\ell_1$-penalized \citet{bierens1990consistent}'s maximum statistic for hypothesis testing on $\theta_0$ satisfying model \eqref{eq:CM model}.
Our penalized inference method is shown to be asymptotically valid when the null hypothesis is true and can be 
calibrated to optimize the asymptotic power against a set of $n^{-1/2}$-local alternatives of interest. 
The penalization tuning parameter is selected by solving a data-dependent max-min problem.
Specifically, we  elaborate on the choice of the penalty parameter $\lambda$ under
a limit of experiments, formally define optimal $\lambda$,
and establish consistency of our proposed calibration method.

Our new test can be viewed as a generalization of the original Bierens' max test with $\lambda=0$ to allow for $\lambda \geq 0$. 
Most of the existing theoretical literature has simply ignored the importance of the choice of $\Gamma$ in maximizing finite sample power of Bierens' max test, although \citet{bierensWang2012integrated} discuss it for the integrated conditional moment test. Intuitively, a larger $\Gamma$ will not reduce the finite-sample power of the test but will increase computational costs. Conversely, a smaller $\Gamma$ may not incur high computational costs but could have an adverse effect on the finite-sample power.
Our test is to let $\Gamma$ be a large region so that it is not a binding constraint and then to choose a scalar $\lambda\geq 0$ in a data-driven way to maximize power. 
To get the same empirical power, it is much easier to optimize over a scalar $\lambda$ than over a binding set $\Gamma \subset \mathbb{R}^p$ when $p$ is boundedly large. 
That is, our penalized test statistic is easier to compute than the one without penalization.\footnote{In convex optimization, a penalized approach is equivalent to a constrained one: that is, penalized optimization with parameter $\lambda$ corresponds to constrained optimization with a different tuning parameter \citep[see, e.g.,][]{boyd2004convex}. A smaller $\lambda$ expands the feasible parameter space in the constrained problem, increasing computational time. For non-convex problems like ours, the equivalence is not guaranteed, but we speculate a similar phenomenon 
because the $\ell_1$ penalty term is convex and a larger $\lambda$ leads the overall objective function to behave more akin to a convex function. To confirm, we analyzed computation times and present the results in Table~\ref{tab:computation-times}.
}
The computational gains by penalization are practically important since the $p$-value is constructed by  a multiplier bootstrap procedure.

We demonstrate the usefulness of our method by applying it to a couple of empirical examples. 
First, as our main example, we revisit \citet{yogo2004} and find that an uninformative confidence interval
(resulting from unconditional moment restrictions)  for the elasticity of intertemporal substitution, based on annual US series ($n \approx 100$ and $p = 4$), can turn into an informative one. 
We provide further supporting evidence via extensive Monte Carlo experiments that mimic \citet{yogo2004}.
Second, as our supplementary example, in the online appendix, we revisit the test of \citet{SSA:2004} and show that our method yields a rejection of the null hypothesis of rational unbiased reporting of ability status at the conventional level even with a small sample size $n \approx 350$ and relatively large $p=21$. Interestingly, \citet{SSA:2004} already implemented their version of Bierens' test and failed to reject the null. This is also consistent with an anonymous referee's point that the original Bierens' tests are not recommended for applications with small sample sizes. Both empirical examples suggest that there is substantive evidence for the efficacy of our proposed method.

\paragraph{Closely related literature}

The original Bierens' statistics \citep{Bierens:82,bierens1990consistent} are designed for consistent specification tests of a parametric null regression model against nonparametric regression alternatives. See, e.g., \citet{deJong,Andrews:97,ICM:1997,Stinchcombe:White,Chen:Fan,fan2000consistent,Lavergne:2008} for earlier various extensions of Bierens' consistent tests of parametric, semiparametric null against nonparametric alternatives in regression settings. 
\citet{Horowitz:2006} is perhaps the first to extend \citet{Bierens:82}'s integrated conditional moment (ICM) statistic to test a null of a parametric IV regression $E[Y-f(X,\theta_0)|W]=0$ (with $X\neq W$) against a nonparametric IV regression alternative. Our motivation is different as we are concerned about testing the null hypothesis regarding $\theta_0$ in \eqref{eq:CM model}, assuming that the underlying IV model \eqref{eq:CM model} is correctly specified while some of the conditioning instruments $W_i$ could be irrelevant. 


Recently, \citet{Antoine:Lavergne} leveraged \citet{Bierens:82}'s ICM statistic to develop an inference procedure for a finite-dimensional parameter of interest within a linear IV regression framework. The first arXiv version of our work is contemporaneous with theirs and shares the same objective of obtaining more informative confidence sets for $\theta_0$ using conditional instruments, accommodating weak IV scenarios. Importantly, our tests are based on \citet{bierens1990consistent}'s maximum statistic, whereas the tests of \citet{Antoine:Lavergne} are built on \citet{Bierens:82}'s ICM statistic. Thus, our tests are complements rather than serving as substitutes, proposing distinct statistics with differing power properties. However, our tests demonstrate superior size and power properties compared to those of \citet{Antoine:Lavergne}, particularly as the dimension of irrelevant instruments increases.\footnote{Detailed results of Monte Carlo experiments can be found in an earlier version of this paper, available at \url{https://arxiv.org/pdf/2008.11140v4}.}

The remainder of the paper is organized as follows.
In Section \ref{sec:test}, we define the test statistic and describe how to obtain bootstrap $p$-values.
Section \ref{sec:bootstrap} establishes bootstrap validity.
In Section \ref{sec:Choice-of-Penalty}, we derive consistency and local power
and  propose how to calibrate the penalization  parameter to optimize the power of the test.
In Section \ref{sec:implementation}, we summarize our proposed inference procedure and provide pseudo-code.
In Section \ref{sec:subvector}, we extend our method to inference for $\theta_{1,0}$,
 for which we use plug-in estimation of $\theta_{2,0}$, where $\theta_0 = (\theta_{1,0}, \theta_{2,0})$.
In Sections~\ref{sec:emp2} and ~\ref{sec:mc}, respectively, we present our main empirical example and Monte Carlo experiments based on \citet{yogo2004}.
Section \ref{sec:concl} gives concluding remarks.
All the proofs are in Section \ref{sec:proofs}.
Online Appendix \ref{sec:emp1} gives our supplementary empirical example using data from \citet{SSA:2004}.

\section{Test Statistic}\label{sec:test}

In this section, we introduce our test statistic and describe how to carry out bootstrap inference.

Before we present our test statistic, we first assume the following conditions.

\begin{assumption}\label{assum0}
\begin{enumerate}[(i)]
\item $\Theta$ is a compact convex nonempty subset of $\mathbb{R}^d$.
\item \label{assum1}
The time series $\left\{ X_{i},W_{i}\right\} $ is strictly stationary and ergodic, 
where $W_{i}$ is adapted to the natural filtration $\mathcal{F}_{i-1}$ up to time $ i-1 $,
 and $\{U_{i}:=g(X_i,\theta_0) \}$ is a martingale difference sequence (mds). Furthermore, $ \mathbb{E} [ |g(X_i,\theta)|^c ] $ is bounded on $\Theta$
 for some $ c> \max\{2,d\}$.
 
\item $\mathbb{E} \left[ g\left(X_i,\theta \right)|W_i \right ] =0 \; a.s.$ if and only if $\theta=\theta_{0}$, where $\theta_0 \in \Theta$.
\item The function $\mathbb{E} \left[ | g\left(X_i,\theta\right) |^2 \right] $ is bounded and bounded away from zero on $\Theta$. 
\item  $W_i$ is a bounded random vector in $\mathbb{R}^p$.
\end{enumerate}
\end{assumption}


The boundedness assumption on $W_i$ is without loss of generality since we can take a one-to-one transformation to ensure that each component of $W_i$ is bounded (for instance, $x \mapsto \tan^{-1}(x)$ componentwise, as used in \citet{bierens1990consistent}).
Condition \eqref{assum1} is standard and ensures the weak convergence of  the stochastic processes $ \sqrt{n}M_n (\gamma) $ and $ s_n^2(\gamma) $, which will be introduced later in this section.



Let $\Gamma$ denote a compact nonempty subset in $\mathbb{R}^{p}$.
Define
\begin{align}\label{M:def}
M\left(\theta,\gamma\right):= \mathbb{E} \left[ g\left(X_i,\theta\right) \exp(W_i'\gamma) \right].
\end{align}
\citet{bierens1990consistent} established the following result. 

\begin{lem}[\citet{bierens1990consistent}]\label{Lem:Bierens}
Let Assumption \ref{assum0} hold.
Then, $M\left(\theta,\gamma\right)=0$ if $\theta=\theta_{0}$
and $M\left(\theta,\gamma\right)\neq 0 \ for \ almost \ every \ \gamma \in \Gamma $ if
$\theta\neq\theta_{0}$. 
\end{lem}


To minimize notational complexity,   we often abbreviate
$M\left(\gamma\right):=M\left(\theta_{0},\gamma\right)$ as for  
$U_{i}:=g\left(X_{i},\theta_{0}\right)$ in condition \eqref{assum1} throughout this paper.
In order to test a hypothesis \[ H_0: \theta_0 =\bar{\theta} \] 
against its negation $H_1: \theta_0 \neq \bar{\theta}$,
we construct a test statistic as follows. 
First, define
\begin{align}\label{test-components}
\begin{split}
{M}_{n}(\theta,\gamma) & :=\frac{1}{n}\sum_{i=1}^{n}g(X_i,\theta)\exp(W_{i}'\gamma),\\
{s}_{n}^{2}(\theta,\gamma) & :=\frac{1}{n}\sum_{i=1}^{n}\left[g(X_i,\theta)\exp(W_{i}'\gamma)\right]^{2},\\
{Q}_{n}(\theta,\gamma) & :=\sqrt{n}\frac{|{M}_{n}(\theta,\gamma)|}{{s}_{n}(\theta,\gamma)},
\end{split}
\end{align}
where we let $Q_n = 0$ if   $s_n = 0$ for a given $\theta$ and $\gamma$. Note that $g(X_i,\theta_{0}) \exp(W_{i}'\gamma)$ is a centered random variable and the criterion $Q_n$ is motivated by the $t$-statistic and it is asymptotically folded standard Normal at $\theta = \theta_0 $ for each $\gamma$. Define the test statistic for the null hypothesis as
\begin{align}
{T}_{n}(\bar{\theta},\lambda) & :=\sup_{\gamma \in \Gamma}\left[{{Q}_{n}(\bar{\theta},\gamma)}-\lambda \|\gamma\|_{1}\right],\label{test-stat-def}
\end{align}
where $\|a\|_{1}$ is the $\ell_{1}$ norm of a vector $a$
and $\lambda \geq 0$ is the penalization parameter.\footnote{Instead of a weighted version of the $\ell_1$-norm, we focus on the simple case for brevity.}
We regard ${T}_{n}(\bar{\theta},\lambda)$ as a stochastic process indexed by $\lambda \in \Lambda$, where $\Lambda$ is a compact nonempty subset in $\mathbb{R}_{+}  := \{\lambda \in \mathbb{R} | \lambda \geq 0 \}$.

We note that our test can be viewed as a generalization of the original Bierens' max test with $\lambda=0$ to allow for $\lambda \geq 0$. 
Under our Assumption \ref{assum0}(iv), one can let $\Gamma = \prod_{j=1}^p [-a_j,a_j]$ for any $a_j > 0,~ j = 1,...,p$. 
The motivation of our test is to let $\Gamma$ be a large region so that $\Gamma$ is not a binding constraint and then to choose a scalar $\lambda\geq 0$ in a data-driven way to maximize power. It is much easier to optimize over a scalar $\lambda$ then over a set of tuning parameters $a_j > 0,~ j = 1,...,p$ when $p$ is large.

We end this section by commenting that our test statistic  ${T}_{n}(\bar{\theta},\lambda)$
is substantially different from the LASSO's criterion. 
First of all, the term ${M}_{n}(\bar{\theta},\gamma)$ is not the least squares objective function;
second, $\gamma$ is different from regression coefficients.  
Furthermore,  if we mimic LASSO more closely, the test statistic would be 
\begin{align}
{T}_{n, \text{alt}}(\bar{\theta},\lambda) & :=
\sup_{\gamma \in \Gamma}
\Big\{ n [{M}_{n}(\bar{\theta},\gamma)]^{2}  -\lambda \|\gamma\|_{1} \Big\}.\label{test-stat-def-alt}
\end{align}
We have opted to consider ${T}_{n}(\bar{\theta},\lambda)$ in the paper because
it is properly studentized and  comparable in scale to  the $\ell_1$-penalty term, unlike 
${T}_{n, \text{alt}}(\bar{\theta},\lambda)$ in \eqref{test-stat-def-alt}.  
In addition, it would be easier to specify $\Gamma$ with ${T}_{n}(\bar{\theta},\lambda)$
because the magnitude of ${M}_{n}(\bar{\theta},\gamma)$ tends to get larger as the scale of $\gamma$ increases 
but ${Q}_{n}(\bar{\theta},\gamma)$ may not. 

Finally, we conclude this section by commenting that our choice of the $\ell_1$-penalty is on an ad hoc basis. 
As is well known for LASSO, the $\ell_1$-penalty is more likely to exclude irrelevant instruments than the $\ell_2$-penalty and shrinks the unpenalized optimizer by the same amount while the $\ell_2$-penalty shrinks it proportionally. 
Also, the penalty term affects the test statistic differently under the null and the alternative hypotheses. 
Ideally, the best penalty should reduce the value of the test statistic more under the null than under the alternatives to maximize the power, while maintaining some computational gains. This issue of choosing an optimal penalty term is a very challenging topic, which deserves further independent research.

\subsection{Bootstrap Critical Values}\label{subsec:bootstrap}


We consider the multiplier bootstrap to carry out inference. Define
\begin{align}\label{mult-bts}
\begin{split}
{M}_{n,\ast}(\theta,\gamma) & :=\frac{1}{n}\sum_{i=1}^{n}\eta_{i}^{\ast}g(X_i,\theta)\exp(W_{i}'\gamma),\\
{s}_{n,\ast}^{2}(\theta,\gamma) & :=\frac{1}{n}\sum_{i=1}^{n}\left[\eta_{i}^{\ast}g(X_i,\theta)\exp(W_{i}'\gamma)\right]^{2},\\
{Q}_{n,\ast}(\theta,\gamma) & :=\sqrt{n} \frac{|{M}_{n,\ast}(\theta,\gamma)|}{{s}_{n,\ast}(\theta,\gamma)}, \\
{T}_{n,\ast}(\theta,\lambda) & :=\sup_{\gamma \in \Gamma}\left[{{Q}_{n,\ast}(\theta,\gamma)}-\lambda \|\gamma\|_{1}\right],
\end{split}
\end{align}
where $\eta_{i}^{\ast}$ is drawn from $N(0,1)$ and independent from
data $\{(X_{i},W_{i}):i=1,\ldots,n\}$.\footnote{Instead of using ${s}_{n,\ast}^{2}(\theta,\gamma)$, we may employ ${s}_{n}^{2}(\theta,\gamma)$ to define ${Q}_{n,\ast}(\theta,\gamma)$. The resulting first-order asymptotic theory would be equivalent.} For each bootstrap replication
$r$, let
\begin{align}
{T}_{n,\ast}^{(r)}(\theta,\lambda) & :=\sup_{\gamma \in \Gamma}\left[{{Q}_{n,\ast}^{(r)}(\theta,\gamma)}-\lambda \|\gamma\|_{1}\right].\label{test-stat-def-b}
\end{align}
For each $\lambda$, the bootstrap $p$-value is defined as
\begin{align*}
{p}_{\ast}(\bar{\theta}, \lambda) :=\frac{1}{R}\sum_{r=1}^{R}1\{ {T}_{n,\ast}^{(r)}(\bar{\theta},\lambda) > {T}_{n}(\bar{\theta},\lambda) \}
\end{align*}
for a large $R$. We reject the null hypothesis at the $\alpha$ level
if and only if ${p}_{\ast}(\bar{\theta},\lambda) < \alpha$.
Then, a bootstrap confidence interval for $\theta_0$ can be constructed by inverting
a pointwise test of $H_0: \theta_0 = \bar{\theta}$.

\section{Bootstrap Validity}\label{sec:bootstrap}

We now introduce some additional notation. Let
$K (\theta,\gamma_1, \gamma_2) :=\mathbb{E} \left[ g(X_i,\theta)^{2}\exp\left(W_i'\left(\gamma_{1}+\gamma_{2}\right)\right) \right]$
and $s^2\left(\theta,\gamma\right) := \mathbb{E} \left[ g(X_i,\theta)^{2}\exp\left(2W_i'\gamma\right) \right]$. 
As before, we suppress the dependence on $ \theta $ when they are  evaluated at $ \theta_{0} $.
The boundedness of $ W_i $ and $ \Gamma $ and the moment condition for $ U_i $ together imply that
\begin{align*}
\sup_{(\gamma_1, \gamma_2) \in \Gamma^2}  K (\gamma_1, \gamma_2) < \infty 
\; \text{ and } \; \inf_{\gamma\in\Gamma}s^2 \left(\gamma\right)>0.    
\end{align*}

 Let $\{ \mathcal{M}\left(\theta,\gamma\right) : \gamma \in \Gamma\}$
be a centered Gaussian process with the covariance kernel
$\mathbb{E} \left[ \mathcal{M}\left(\theta,\gamma_{1}\right)\mathcal{M}\left(\theta,\gamma_{2}\right) \right]
=K (\theta,\gamma_1, \gamma_2)$.
Also, let $ \cvw $ denote the weak convergence in the space of uniformly bounded functions on the parameter space that is  endowed with the uniform metric. Additionally, we define $L_\infty(A)$ as the space of all bounded functions defined on $A = \Gamma$ or $A = \Lambda$.

We first establish the weak convergence of ${T}_{n}(\lambda)$.

\begin{thm}\label{thm:fixed p null}
Let Assumptions \ref{assum0} hold.
Then,
\begin{align}
\sqrt{n}{M}_{n}(\gamma) & \cvw\mathcal{M}\left(\gamma\right) \ \text{ in } L_\infty(\Gamma), \label{eq:FCLT}\\
{s}_{n}(\gamma) & \cvp s\left(\gamma\right) \quad\text{uniformly in }\Gamma.\label{eq:ULLN}
\end{align}
Furthermore,
\[
{T}_{n}(\lambda) \cvw {T}(\lambda) :=\sup_{\gamma\in\Gamma}\left[\frac{\left|\mathcal{M}\left(\gamma\right)\right|}{s\left(\gamma\right)}-\lambda \|\gamma\|_{1}\right] \ \text{ in }L_\infty(\Lambda).
\]
\end{thm}

We now show that the bootstrap analog ${T}_{n,*}(\lambda)$ of ${T}_{n}(\lambda)$ converges weakly to the same limit. 
The definition of the conditional weak convergence, $ \cvwst $ in $ P $, and conditional convergence in probability, $ \cvpst $  in $ P $, employed in the following theorem is referred to e.g. Section 2.9 in \cite{VW:1996:book}.

\begin{thm}\label{thm:fixed p null:bootstrap}
Let Assumption \ref{assum0} hold.
Then, for any fixed $ \bar{\theta } $,
\begin{align}
\sqrt{n}{M}_{n,*}( \bar{\theta },\gamma) & \cvwst\mathcal{M}\left(\bar{\theta },\gamma\right)\ \text{ in } L_\infty(\Gamma) \quad \text{ in } P, \label{eq:FCLT-1}\\
{s}_{n,*}(\bar{\theta },\gamma) & \cvpst s\left(\bar{\theta },\gamma\right)  \quad \text{uniformly in }\Gamma \quad \text{ in } P. \label{eq:ULLN-1}
\end{align}
Furthermore,
\[
{T}_{n,*}(\bar{\theta },\lambda) \cvwst {T}(\bar{\theta },\lambda) :=\sup_{\gamma\in\Gamma}\left[\frac{\left|\mathcal{M}\left(\bar{\theta },\gamma\right)\right|}{s\left(\bar{\theta },\gamma\right)}-\lambda \|\gamma\|_{1}\right]\ \text{ in } L_\infty(\Lambda) \quad \text{in } P.
\]
\end{thm}

Recall that we abbreviate
${M}_{n}(\gamma) := {M}_{n}(\theta_0,\gamma)$,
${s}_{n}^{2}(\gamma):={s}_{n}^{2}(\theta_0,\gamma)$ and ${T}_{n}(\lambda):={T}_{n}(\theta_0,\lambda)$.
While Theorem \ref{thm:fixed p null} holds at $\theta = \theta_0$, Theorem \ref{thm:fixed p null:bootstrap} holds for any $ \bar{\theta } $.
Therefore, Theorems \ref{thm:fixed p null} and \ref{thm:fixed p null:bootstrap} imply that the bootstrap critical values are valid for any drifting sequence of $ \lambda_n \in \Lambda$ under some smoothness condition on the limiting distribution and thus for the calibration of $\Lambda$ described in Section~\ref{sec:lambda:calibration}. 
See the proof of Theorem \ref{thm:consistency of lambda} for more details.

\section{Consistency, Local Power and Calibration of \texorpdfstring{$\lambda$}{lambda}}\label{sec:Choice-of-Penalty}



\subsection{Consistency}

Suppose that the null hypothesis is set as $H_{0}:\theta_0 =\bar{\theta}$ for some $\bar{\theta}\neq\theta_{0}$. Then, being explicit about the null, we write
\begin{align*}
{M}_{n}(\bar{\theta},\gamma) & =\frac{1}{n}\sum_{i=1}^{n}g\left(X_{i},\bar{\theta}\right)\exp(W_{i}'\gamma)\cvp \mathbb{E} \left[ g\left(X_i,\bar{\theta}\right)\exp(W_i'\gamma) \right],\\
{s}_{n}(\bar{\theta},\gamma) & \cvp\sqrt{
\mathbb{E} \left[ g^2 \left(X_i,\bar{\theta}\right)\exp(2W_i'\gamma) \right]
}.
\end{align*}
Therefore,  for any $\lambda \in \Lambda$,
\[
n^{-1/2} T_n (\bar{\theta},\lambda) \cvp
\sup_{\gamma \in \Gamma}
\frac{\left| \mathbb{E} \left[ g\left(X_i,\bar{\theta}\right)\exp(W_i'\gamma) \right] \right|}
	{\sqrt{\mathbb{E} \left[ g^2 \left(X_i,\bar{\theta}\right)\exp(2W_i'\gamma) \right] }}.
\]
On the other hand, the bootstrap statistic is always $ O_p (1) $,
yielding the consistency of the test based on $ T_n (\bar{\theta},\lambda) $, as in the following theorem.

\begin{thm}\label{thm:consistency}
Let Assumptions \ref{assum0} hold.
%
Then, for $ \bar{\theta} \neq \theta_{0} $,
${T}_{n} (\bar{\theta}, \lambda) \cvp +\infty$ for any $\lambda \in \Lambda$.
\end{thm}

\subsection{Local Power}
\label{sec:local-power}

Consider a sequence of local hypotheses of the following form: for some nonzero constant vector $B$,
\[
\theta_{n}:=\theta_{0}+B \, n^{-1/2},
\]
which leads to the following leading term after linearization: 
\begin{align}\label{local-power-seq}
U_{i,n}:=U_i +   G\left(X_{i},\theta_{0}\right)  B \, n^{-1/2},
\end{align}
where $G\left(X_{i},\theta\right):=\partial g\left(X_{i},\theta\right)/\partial\theta'$,
assuming the continuous differentiability of $g\left( \cdot \right)$ at $ \theta_{0} $.
Unless $g\left(X_{i},\theta\right)$ is linear in $\theta$, the term $G\left(X_{i},\theta_0\right)$ depends on $\theta_0$.
However, under the null hypothesis,  $G\left(X_{i},\theta_{0}\right)$ is completely specified.
For $B$, we may set it as, for example, a vector of ones times a constant.
The form of \eqref{local-power-seq} will be intimately related to our proposal regarding how to
calibrate the penalization parameter $\lambda$.

As before, write $G_{i}:=G\left(X_{i},\theta_{0}\right)$.
Under \eqref{local-power-seq}, we have
\begin{align*}
\sqrt{n}{M}_{n}(\theta_{n},\gamma) & =\frac{1}{\sqrt{n}}\sum_{i=1}^{n}U_{i}\exp(W_{i}'\gamma)+\frac{1}{n}\sum_{i=1}^{n}\exp(W_{i}'\gamma) G_i B +o_{p}\left(1\right) \\
 & \Rightarrow\mathcal{M}\left(\gamma\right)+
 \mathbb{E} \left[ \exp(W_i'\gamma)G_i B \right]\ \text{ in } L_\infty(\Gamma).
\end{align*}
Then, we can establish that
\begin{equation}
{T}_{n}(\theta_n,\lambda) \cvw \sup_{\gamma \in \Gamma}\left[\left|\frac{\mathcal{M}\left(\gamma\right)+
\mathbb{E} \left[ \exp(W_i'\gamma)G_i B \right]}{s\left(\gamma\right)}\right|-\lambda\left\Vert \gamma\right\Vert _{1}\right]\ \text{ in }L_\infty(\Lambda),\label{eq:local limit}
\end{equation}
using  arguments identical to those to prove Theorem~\ref{thm:fixed p null}.

Define the noncentrality term
\begin{align*}
\kappa\left(\gamma,B\right):=\frac{\mathbb{E} \left[ \exp(W_i'\gamma)G_i B \right]}
{s(\gamma)}=\frac{\mathbb{E} \left[ \exp(W_i'\gamma)G_i B \right]}
{\sqrt{\mathbb{E} [U_i^{2}\exp(2W_i'\gamma)]}}.
\end{align*}
Define $\mathcal{Q}(\gamma) := \frac{\mathcal{M}(\gamma)}{s(\gamma)}$ and $T(\lambda,B) := \sup_{\gamma \in \Gamma} \vert \mathcal{Q}(\gamma) +\kappa(\gamma,B) \vert -\lambda  \left\Vert \gamma\right\Vert _{1} $.  For the test to have a nontrivial power, we need that $\kappa\left(\gamma^{*}(B,\lambda),B\right)\neq0$
with a positive probability, where $\gamma^{*}(B,\lambda)$ denotes a (random)
maximizer of the stochastic process in \eqref{eq:local limit}.
Since the penalty affects $\gamma^{*}(B,\lambda)$ in different ways under the null of $ B=0 $ and alternatives of $ B \neq 0 $, its implication on power of the test is not straightforward to analyze.
The subsequent subsection proposes a method to select $ \lambda $ in a more systematic way to increase power.

We discuss some sufficient conditions, under which the presence of penalty increases the power of the test in the setting of the preceding limit experiment. Heuristically, if the unpenalized criterion is maximized at $\gamma$  near zero, then the maximum would be less affected by the introduction of the $\ell_1$ penalty on $\gamma$.  And a suitable noncentrality function $\kappa$, like e.g. a concave function with a unique maximum at zero, can force the maximizing $\gamma$  closer to zero. That is, the power gain is obtained through the penalization if the limit experiment under the null is maximized at bigger $\gamma$ than under the alternative. The following lemma is a more formal treatment of the heuristic. 

\begin{lem} \label{lem:power gain}
Suppose that there exist a unique $ \tilde{\gamma} (b) $ maximizing $ \vert \mathcal{Q}\left(\gamma\right) +\kappa(\gamma,b) \vert $ for $b=0$ or $B$. If  
\begin{equation}\label{eq:Lambda st H0H1} 
       \left\Vert \tilde{\gamma} (B) \right\Vert_1 <  \left\Vert \tilde{\gamma} (0) \right\Vert_1 \quad a.s. 
  \end{equation}
then, we can find a more powerful test with a strictly positive $\lambda$ than the test with $\lambda=0$ in the limit experiment with $B\neq0$.  
\end{lem}

The assumption that $  \mathcal{Q}\left(\gamma\right) +\kappa(\gamma,b) $ is uniquely maximized a.s. on a compact space is standard due to \citet[Lemma 2.6]{kim1990cube},  requiring only that the increments of the Gaussian process exhibit non-negligible variances. 

An example that may meet the condition (\ref{eq:Lambda st H0H1}) is the case where the noncentrality term $ \kappa(\gamma,B) $ under the alternative hypothesis induces a sparse solution. This happens when the set of instrumental variables $ W $ contains redundant elements. It is similar to the well-known fact that the presence of an irrelevant variable in the linear regression results in loss of power in the tests based on the OLS estimates. 
Specifically, suppose for simplicity that  $W_i=\left(Z_i,F_i\right)$ and $F_i$ is a pure noise that is independent of everything else. Then, the noncentrality
term can be rewritten as
\[
\kappa\left(\gamma,B\right)=
\frac{\mathbb{E} \left[ \exp(Z_i'\gamma_1) G_i B \right]}
{\sqrt{\mathbb{E} [U_i^{2}\exp(2Z_i'\gamma_1)]}}
\frac{\mathbb{E} \left[ \exp(F_i'\gamma_2) \right]}
{\sqrt{\mathbb{E} [\exp(2F_i'\gamma_2)]}}
=:\kappa_{1}\left(\gamma_{1}, B\right)\kappa_{2}\left(\gamma_{2}\right).
\]
Then, Jensen's inequality yields that $\kappa_{2}\left(\gamma_{2}\right) \leq 1$ and the equality holds if and only if $\gamma_{2}=0$. If the dimension of $F_i$ is relatively large compared to that of $Z_i$ and the magnitude of $\kappa_{2}\left(\gamma_{2}\right)$ dominates that of $\kappa_{1}\left(\gamma_{1}, B\right)$ and $\mathcal{M}(\gamma)$  then $\tilde{\gamma}(B)$  would be closer to zero than  $\tilde{\gamma}(0)$, which is a maximizer of a centered Gaussian process of equal marginal variance. 

Since analytical derivation is involved in general cases,  we provide visual representation of the preceding discussion via some Monte Carlo simulation as below. 

Specifically, we generate a random sample with $n=1000$ from a simple linear regression model with normal random variables.
Specifically, we draw independent variables uniformly distributed on $[-1,1]$,  $\varepsilon_{ji}, \ j=1,4$  and independent standard normal variables $\varepsilon_{ji}, \ j=2,3$ and generate $X_i=\varepsilon_{1i} + \varepsilon_{2i}/2$, $Y_i= X_i \theta_n + (\varepsilon_{2i} + \varepsilon_{3i})/2$ with $\theta_{n}=\theta_{0}+B \, n^{-1/2}$ with some nonzero constant $B$, and $Z_i=(W_i,F_i)$, where $ W_i= \varepsilon_{1i} - \varepsilon_{4i}$, while $F_i$ is a $(p-1)$-dimensional independent vector uniformly distributed on $[-1,1]$,  that is independent of all the others. 
Figure \ref{fig-power} plots the ``theoretical'' power functions of our proposed test, where the power curves are obtained via Monte Carlo simulations with $\Gamma = [-5,5]^p$.
There is only one endogenous regressor here and the three lines in the figure represent the power curves as a function of the penalty level $ \lambda $ for three different values of the dimension $ p $ of $ Z_i $, whose first element is strongly correlated to $ U_i $ while the others are irrelevant. Thus, $ p-1 $ represents the number of irrelevant instruments. 
The power decreases as the number of irrelevant variables increases when the penalty $ \lambda = 0 $. This is analogous to the textbook treatment of hypothesis testing with the linear regression model. 
Next, the power increases gradually up to a certain point as the penalty grows for each $p$, in line with the preceding discussion. 
In addition, the power gain from the penalization, that is, the difference between the maximum power and the power at $ \lambda = 0 $, is bigger for larger $p$.

\subsection{Calibration of \texorpdfstring{$\lambda$}{lambda}}\label{sec:lambda:calibration}

The penalty function works differently on how it shrinks the maximizer
$\tilde{\gamma}$ under the alternatives. Ideally, it should induce
sparse solutions that force zeros for the coefficients of the irrelevant
conditioning variable to maximize the power of the test.


Although it is demanding to characterize the optimal choice of $\lambda$ analytically, we
can elaborate on the choice of the penalty parameter $\lambda$ under
the limit of experiments
\[
\frac{\mathcal{M}\left(\gamma\right)+
\mathbb{E} \left[ \exp(W_i'\gamma)G_i B \right]}{s\left(\gamma\right)},
\]
for which we parametrize the size of the deviation by $B$.
Then, our test becomes
\begin{equation}
\mathcal{T}\left(\lambda,B,\alpha\right):=1\left\{ \sup_{\gamma \in \Gamma}\left|\frac{\mathcal{M}\left(\gamma\right)+
\mathbb{E} \left[ \exp(W_i'\gamma)G_i B \right]}{s\left(\gamma\right)}\right|-\lambda\left\Vert \gamma\right\Vert _{1}>c_{\alpha}(\lambda) \right\} \label{eq:calT}
\end{equation}
for a critical value $c_{\alpha}(\lambda)$,  which is  the $(1-\alpha)$ quantile of
$\sup_{\gamma\in\Gamma}\left[\frac{\left|\mathcal{M}\left(\gamma\right)\right|}{s\left(\gamma\right)}-\lambda \|\gamma\|_{1}\right]$.

Let $\mathcal{R}\left(\lambda,B,\alpha\right):=\mathbb{E} \left[ \mathcal{T}\left(\lambda,B,\alpha\right) \right]$
denote the power function of the test under the limit experiment for
given $\lambda$, $B$ and $\alpha$, where  $0 < \alpha < 1$ is a prespecified  level of the test. We propose to select $\lambda$ by solving the max-min problem:
\begin{align}\label{max-min}
\sup_{\lambda \in \Lambda} \inf_{B\in \mathcal{B}}\mathcal{R}\left(\lambda,B,\alpha\right),
\end{align}
where $\Lambda$ is a set of possible values of $\lambda$
and
$\mathcal{B}$ is a set of possible values of $B$.
In some applications, where $B$ is one-dimensional and 
$\left| \mathcal{M}\left(\gamma\right)+
\mathbb{E} \left[ \exp(W_i'\gamma)G_i B \right] \right|$
is stochastically monotone in $|B|$,
the inner minimization over $B \in \mathcal{B}$ is simple and easy to characterize.
For $\Lambda$, we can take a discrete set of possible values of $\lambda$, including 0, if suitable.
The idea behind \eqref{max-min} is as follows.
For each candidate $\lambda$, the size of the test is constrained properly because
$\mathcal{R}\left(\lambda,0,\alpha\right)\leq\alpha$. Then we look at the least-favorable local power among possible values of $B$ and choose $\lambda$ that maximizes the least-favorable local power.

To operationalize our proposal, we again rely on a multiplier bootstrap. Define
\begin{align}\label{mult-bs-B}
\begin{split}
{M}_{n,\ast,B}(\gamma) & :=\frac{1}{n}\sum_{i=1}^{n}
\left ( \eta_{i}^{\ast} U_{i} + \frac{B}{\sqrt{n}} G_i \right) \exp(W_{i}'\gamma),\\
{s}_{n,\ast,B}^{\,2}(\gamma) & :=\frac{1}{n}\sum_{i=1}^{n}\left[\left ( \eta_{i}^{\ast} U_{i} + \frac{B }{\sqrt{n}} G_i \right)\exp(W_{i}'\gamma)\right]^{2},\\
{Q}_{n,\ast,B}(\gamma) & :=\sqrt{n}\frac{|{M}_{n,\ast,B}(\gamma)|}{{s}_{n,\ast,B}(\gamma)},
\end{split}
\end{align}
where $\eta_{i}^{\ast}$ is drawn from $N(0,1)$ and independent from
data $\{(X_{i},W_{i}):i=1,\ldots,n\}$.
The quantities above are just shifted versions of \eqref{mult-bts}.\footnote{%
Notice that we use a shifted version of ${s}_{n,\ast,B}^{2}(\gamma)$
instead of ${s}_{n,\ast,0}^{2}(\gamma)$ when we define \eqref{mult-bs-B}. This is because we would like to mimic more closely the finite-sample distribution of the test statistic under the alternative.}
For each bootstrap replication
$r$, let
\begin{align}\label{test-stat-def-b-nc}
{T}_{n,\ast,B}^{(r)}(\lambda) & :=\sup_{\gamma \in \Gamma}\left[{Q}_{n,\ast,B}^{(r)}(\gamma)-\lambda \|\gamma\|_{1}\right].
\end{align}
Then the critical value $c_{\alpha}(\lambda)$ is approximated by $ c_{\alpha}^*(\lambda) $, the $(1-\alpha)$-quantile of
${T}_{n,\ast,0}(\lambda) \equiv T_{n,*}(\lambda)$. Once $c_{\alpha}^*(\lambda)$
is obtained,
$\mathcal{R}\left(\lambda,B,\alpha\right)$ is approximated by
\begin{align}\label{RP:bootstrap}
\frac{1}{R}\sum_{r=1}^{R}1 \left\{ {T}_{n,\ast,B}^{(r)}(\lambda) >c_{\alpha}^*(\lambda) \right\}.
\end{align}

Similarly, let $ \widehat{\lambda} $ denote a maximizer of $\min_{B\in \mathcal{B}} \mathcal{R}_n (\lambda,B,\alpha) $ over $\Lambda$, where
$$
\mathcal{R}_n (\lambda,B,\alpha) := {\Pr}^{*} \{ {T}_{n,*,B} (\lambda)  > c_{\alpha}^{*} (\lambda)\},
$$
and $ {\Pr}^{*} $ denotes the conditional probability of the bootstrap sample given the sample.
Recall
$ {T}(\lambda) = \sup_{\gamma\in\Gamma}\left[\frac{\left|\mathcal{M}\left(\gamma\right)\right|}{s\left(\gamma\right)}-\lambda \|\gamma\|_{1}\right] $.
Let $ F_{\lambda}  $ and $ F^{*}_{\lambda}  $ denote the distribution function of $ T(\lambda) $ and that of $ T_{n,*}(\lambda) $ conditional on the sample $ \mathcal{X}_n $, respectively.
We make the following regularity condition on $ F_{\lambda}  $.
Let $c_{\alpha}(\lambda)$ denote  the $(1-\alpha)$-quantile of ${T}(\lambda)$ and 
$ A_{\alpha} = \{(x,\lambda): |c_{\alpha}(\lambda)-x| \leq c, \lambda \in \Lambda \}$ for some positive constant $c$.

\begin{assumption}\label{assmp:lambda-cons}
The partial derivative $\partial  F_{\lambda}(x) / \partial x $ is positive and continuous on  $ A_{\alpha} $.
\end{assumption}

When  $\lambda = 0$, $T(\lambda)$ is the maximum of a centered Gaussian process. In view of the well-known anti-concentration property of the maximum of a centered Gaussian process \citep[see, e.g.,][]{chernozhukov2014anti}, it is not restrictive to assume the presence of a density (with respect to Lebesgue measure) of $T$ at $\lambda=0$. Furthermore, the extension to noncentral Gaussian process is given by \citet{chernozhukov2016empirical}. 
Thus, we may assume $T(\lambda)$ has a bounded density for all $\lambda \in \Lambda$. Thus, its Lipschitz property in $\lambda$ implies that its distribution function indexed by $\lambda$ is also Lipschitz in $\lambda$
because 
\begin{multline*}
\Pr \left\{ T\left(\lambda_{2}\right)\leq x\right\} 
= \Pr \left\{ T\left(\lambda_{1}\right)\leq x+T\left(\lambda_{1}\right)-T\left(\lambda_{2}\right)\right\} \\
\leq \Pr \left\{ T\left(\lambda_{1}\right)\leq x+c\left|\lambda_{1}-\lambda_{2}\right|\right\} 
\leq \Pr \left\{ T\left(\lambda_{1}\right)\leq x\right\} +C\left|\lambda_{1}-\lambda_{2}\right|,
\end{multline*}
where constants $c$ and $C$ depend only on the Lipschitz constant and the densities, respectively. 
Then, the continuity assumption in Assumption \ref{assmp:lambda-cons} is reasonable. 

To sum up, we formally define optimal $\lambda$ assuming $\mathcal{R}(\lambda,B,\alpha )$ is continuous on $\Lambda \times \mathcal{B}$, which is possibly set-valued, in the following way.


\begin{definition}
Let $ \Lambda_0 \subset \Lambda$ denote a set of the global solution in \eqref{max-min} so that
\begin{align*}
\min_{B\in \mathcal{B}}\mathcal{R}\left(\lambda_0,B,\alpha\right) \geq \min_{B\in \mathcal{B}}\mathcal{R}\left(\lambda,B,\alpha\right)
\ \ \text{for any $ \lambda \in \Lambda $ and $ \lambda_0 \in \Lambda_0$},
\end{align*}
where the inequality is strict for $\lambda \notin  \Lambda_0$.
\end{definition}

The optimal set $\Lambda_0$ depends on the set $\Lambda$ of possible values of $\lambda$, the set $\mathcal{B}$ of possible values of 
$B$ in  \eqref{local-power-seq},
and the level $\alpha$ of the test.

The following theorem shows that the bootstrap critical values $c_{\alpha}^{*} (\lambda)$ are uniformly consistent for $c_{\alpha}(\lambda)$.
Then, it establishes consistency of our proposed calibration method in the sense that  $  d(\widehat{\lambda}, \Lambda_0) \cvp 0 $, where $ d(x,X) := \inf\{|x-y| : y \in X\} $.

\begin{thm}\label{thm:consistency of lambda}	
	Let Assumptions \ref{assum0} and \ref{assmp:lambda-cons} hold. Then, $c_{\alpha}^{*} (\lambda) \cvp c_{\alpha} (\lambda)$ uniformly in $ \Lambda $.
	Furthermore,
	$ d(\widehat{\lambda}, \Lambda_0) \cvp 0$ and $\Pr \{ T_{n}(\widehat{\lambda})  \geq c^*_{\alpha}(\widehat{\lambda})\} \to \alpha $.
\end{thm}

\section{Implementation}\label{sec:implementation}

In this section, we summarize our inference procedure and provide pseudo-code in Algorithm~\ref{alg:}, which describes how to conduct the pointwise test of $H_0: \theta = \bar{\theta}$.


\RestyleAlgo{ruled}

\SetKwComment{Comment}{/* }{ */}
\SetKwInOut{KwAddIn}{Input for bootstrap test}
\SetKwInOut{KwAddOut}{Auxiliary output}

\begin{algorithm}[htbp!]
\caption{Inference for $H_0:\theta_0 = \bar \theta$ versus $H_1: \theta_0 \ne \bar\theta$ under the conditional moment model $\mathbb E[g(X_i, \theta_0)|W_i] = 0$}
\label{alg:}

\KwIn{the data set $\{(X_i, W_i):i=1,\ldots, n\}$, the hypothesized parameter $\bar \theta$, the search space $\Gamma \subset \mathbb R^p$, the grid for penalty levels $\Lambda \subset \mathbb R_+$. 
}


\For{$k = 1,2,\ldots, p$}{

    Calculate $\bar W_k = \frac{1}{n} \sum_{i=1}^n W_{ki}$ and $s_n(W_k) = \frac{1}{n-1} \sum_{i=1}^n (W_{ki} - \bar W_k)^2$.

    Transform $W_{ki} \leftarrow \tan^{-1}\left( (W_{ki} - \bar W_k)/s_n(W_{k})\right)$.
}

Declare variables $U_i \leftarrow g(X_i, \bar \theta)$.

Declare a function $\gamma \mapsto Q_n(\bar \theta, \gamma)$ according to \eqref{test-components}.

\For{each $\lambda$ in $\Lambda$}{
Calculate $T_n(\bar \theta, \lambda) = \sup_{\gamma \in \Gamma} \left[ Q_n(\bar \theta, \gamma) - \lambda \|\gamma\|_1 \right]$.
}

\KwOut{the set of test statistics $\{T_n(\bar \theta, \lambda): \lambda \in \Lambda \}$.}

\KwAddIn{the number of bootstrap replications $R$, the set of local alternatives $\mathcal B \subset \mathbb R^d$, the significance level $\alpha$.}

\For{$r = 1,2,\ldots, R$}{

Generate $\{\eta_i^*:i=1,\ldots, n\}$ from i.i.d. $N(0, 1)$ independently of the data.

\For{each $\lambda$ in $\Lambda$}{
Declare a function $\gamma \mapsto Q_{n, *}(\bar \theta, \gamma)$ according to \eqref{mult-bts}.

Calculate $T^{(r)}_{n,*}(\bar \theta, \lambda) = \sup_{\gamma \in \Gamma} \left[ Q_{n, *}(\bar \theta, \gamma) - \lambda \|\gamma\|_1 \right]$.

\For{each $B$ in $\mathcal B$}{

Declare a function $\gamma \mapsto {Q}_{n, *, B}(\gamma)$ according to \eqref{mult-bs-B}.

Calculate ${T}^{(r)}_{n,*,B}(\lambda) = \sup_{\gamma \in \Gamma} \left[ {Q}_{n, *,B}(\gamma) - \lambda \|\gamma\|_1 \right]$.

}

}

}

\For{each $\lambda$ in $\Lambda$}{
Declare $cv^*_\alpha(\lambda) \leftarrow (1-\alpha)\text{-quantile of }\{T_{n,*}^{(r)}(\bar \theta, \lambda):r=1,\ldots, R\}$.

\For{each $B$ in $\mathcal B$}{
Compute the simulated local power as 
$\mathcal R_n(\lambda, B, \alpha) = \frac{1}{R} \sum_{r=1}^R 1\{ {{T}}_{n, *, B}^{(r)}(\lambda) > cv^*_\alpha(\lambda)\}$.
}
}

Select the optimal penalty level as $\hat \lambda(\bar \theta) \in \arg \max_{\lambda \in \Lambda} \min_{B \in \mathcal B} \mathcal R_n(\lambda, B, \alpha).$

Compute the p-value as $p_*(\bar \theta, \hat \lambda(\bar \theta)) = \frac{1}{R} \sum_{r=1}^R 1\{T_{n, *}^{(r)}(\bar \theta, \hat \lambda(\bar \theta)) > T_n(\bar \theta, \hat \lambda(\bar \theta))\}$.

\KwOut{Reject $H_0$ at the significance level $\alpha$ iff $p_*(\bar \theta, \hat \lambda(\bar \theta)) < \alpha$.}


\end{algorithm}

In addition to the usual input such as the confidence level $\alpha$ and the number of bootstrap replications $R$, we need to specify the search space $\Gamma \subset \mathbb R^p$, the grid for penalty levels $\Lambda \subset \mathbb R_+$, and the set of local alternatives $\mathcal B \subset \mathbb R^d$.
In our numerical work, we choose $\Gamma = [-a, a]^p$ with some constant $a$ (e.g., $a=1,5$).
For $\Lambda$, we recommend excluding $\lambda = 0$ if $p$ is somewhat large (e.g., $p > 5, 10$).
Regarding $\mathcal B$, it is necessary to know the structure of the inference problem in hand. In our Monte Carlo experiments as well as empirical applications, $d=1$ and we need to choose $\mathcal B$ as a subset of $\mathbb{R}$. We provide details in Section~\ref{sec:emp2}. 

We now make several remarks on computation of $T_n(\bar \theta, \lambda)$, $T^{(r)}_{n,*}(\bar \theta, \lambda)$, and 
${T}^{(r)}_{n,*,B}(\lambda)$ in Algorithm~\ref{alg:}.
For the first empirical application to \citet{yogo2004} data set in Section~\ref{sec:emp2}, 
we use both the grid search (GS) and the \texttt{particleswarm} particle swarm optimization (PSO) solver available in \texttt{Matlab}.\footnote{Specifically, the \texttt{particleswarm} solver is included in  the global optimization toolbox software of \texttt{Matlab}.} We find our tests computed using GS and PSO on the set $\Gamma =[-a,a]^{4}$ perform similarly.
For the second empirical application in Section~\ref{sec:emp1}, we only use the \texttt{particleswarm} for optimization over $\Gamma=[-a,a]^{21}$. PSO is a stochastic population-based optimization method proposed by \citet{kennedy1995particle}.
It conducts gradient-free  global searches  and has been successfully used in economics \citep[for example, see][]{Qu:Tkachenko:2016}.\footnote{It is possible to adopt
the two-step approach used in \citet{Qu:Tkachenko:2016}. That is, we start with the PSO solver, followed by multiple local searches. Further, the genetic algorithm (GA) can be used in the first step instead of PSO and both GA and PSO methods can be compared to check whether a global solution is obtained. We do not pursue these refinements to save the computational times of bootstrap inference.}
Hence, PSO can be viewed as a more refined approach to global optimization than simple grid search.
For the Monte Carlo experiments that mimic \citet{yogo2004} data set in Section~\ref{sec:mc}, we only apply GS on the set $\Gamma =[-a,a]^{p},~p=4,6$. This is 
because PSO is based on a heuristic procedure and requires careful monitoring to check whether it produces reasonable solutions. Furthermore, it is easy to vectorize using GS but harder using the \texttt{particleswarm} solver. In short, it was too costly to monitor the \texttt{particleswarm} solver in the Monte Carlo experiments; however, it was possible with empirical applications because we did not have to regenerate data.

We end this section by recalling that the confidence interval for $\theta_0$ can be constructed by inverting a pointwise test of $H_0: \theta_0 = \bar{\theta}$.
For this purpose, one could generate $R$ collections of $\{\eta_i^*:i=1,\ldots, n\}$ and use the same collections across different values of $\bar{\theta}$ to reduce the random noise in bootstrap inference.

\section{Inference with Pre-Estimated Parameters}\label{sec:subvector}

Partition $\theta=\left(\theta_{1}',\theta_{2}'\right)'$
and
$\theta_0=\left(\theta_{1,0}',\theta_{2,0}'\right)'$.
We now consider
inference for $\theta_{1,0}$. We assume that
for each $\theta_1$,
there exists a preliminary
estimator $\widehat{\theta}_{2} (\theta_1) = \psi_{n}\left( \{X_{i},W_{i} \}_{i=1}^n \right)$ of $\theta_{2}(\theta_{1})$,
so that $\theta_{2,0} = \theta_{2}(\theta_{1,0})$.
 For example, suppose that
$g \left(X_i,\theta \right)$ can be written as
$g \left(X_i,\theta \right) = g_1 \left(X_i,\theta_1 \right) - \theta_2$.
Then, $\theta_{2,0} = \mathbb{E} \left[ g_1 \left(X_i,\theta_{1,0} \right) \right]$, thereby yielding the following estimator of $\theta_2$ given $\theta_1$:
\begin{equation}
\widehat{\theta}_{2} (\theta_1) = n^{-1} \sum_{i=1}^n g_1 \left(X_i,\theta_1 \right). \label{eq:theta2hat}
\end{equation}

In what follows, we assume standard regularity conditions on $\widehat{\theta}_{2} (\theta_1)$.
Let $\Theta_1$ denote the parameter space for $\theta_1$, which is a compact set with a non-empty interior. 
Let $\| a \|$ denote the Euclidean norm of a vector $a$.

\begin{assumption}
\label{ass:expansion}Suppose that
there exists a $\sqrt{n}$-consistent
estimator $\widehat{\theta}_{2} (\theta_1)$ of $\theta_{2}(\theta_1)$ that has the following
representation: uniformly in  $\theta_1 \in \Theta_1$, which is compact and has a non-empty interior,
\[
\sqrt{n}\left(\widehat{\theta}_{2}(\theta_1) -\theta_{2}(\theta_1) \right)=
\frac{1}{\sqrt{n}}\sum_{i=1}^{n}\zeta_{ni}(\theta_1) +o_{p}\left(1\right),
\]
where $\left\{ \psi_{ni}(\theta_1) = (\zeta_{ni}(\theta_1) , U_i) \right\} $ is a strictly stationary ergodic mds array and $$V_{\psi (\theta_1)}:=\lim_{n\to\infty}
\mathbb{E} [ \psi_{ni}(\theta_1)\psi_{ni}(\theta_1)' ]$$ is positive definite.
Furthermore, assume that there exists $G(x, \theta)$ such that
$ \mathbb{E} \|G(X_i,\theta_{0})\| <\infty $, $\theta \mapsto G(X_i,\theta)$ is  continuous at $\theta_{0}$ almost surely, and
\begin{equation}
g\left(X_{i},\theta\right)-g\left(X_{i},\theta_{0}\right)-G\left(X_{i},\theta_0\right)'\left(\theta-\theta_{0}\right)=o_{p}\left(\left\Vert \theta-\theta_{0}\right\Vert \right), \label{eq:expansion of g}
\end{equation}
and that $\theta_1 \mapsto \theta_{2}(\theta_1)$ is continuously differentiable.
\end{assumption}

Define $\widehat{U}_{i} (\theta_{1}):=g [X_{i}, \{\theta_{1},\widehat{\theta}_{2}(\theta_1) \} ]$,
$\widehat{U}_{i}:=\widehat{U}_{i} (\theta_{1,0})$. We introduce the following demeaned statistics\footnote{The test statistic $\widehat{Q}_{n}(\theta_{1},\gamma)$ is effectively based on an orthogonalized residual after estimating the unconditional mean.}
\begin{align*}
\widehat{M}_{n}(\theta_{1},\gamma) & :=\frac{1}{n}\sum_{i=1}^{n}\widehat{U}_{i}(\theta_{1})\left(\exp(W_{i}'\gamma) - \frac{1}{n}\sum_{j=1}^n \exp(W_{j}'\gamma)\right),\\
\widehat{s}_{n}^{2}(\theta_{1},\gamma) & :=\frac{1}{n}\sum_{i=1}^{n}\left[\widehat{U}_{i}(\theta_{1})\left(\exp(W_{i}'\gamma) - \frac{1}{n}\sum_{j=1}^n \exp(W_{j}'\gamma)\right)\right]^{2},\\
\widehat{Q}_{n}(\theta_{1},\gamma) & :=\sqrt{n}\frac{|\widehat{M}_{n}(\theta_{1},\gamma)|}{\widehat{s}_{n}(\theta_{1},\gamma)},
\end{align*}
and the test statistic
\begin{align}
\widehat{T}_{n}(\theta_{1},\lambda) & :=\sup_{\gamma \in \Gamma}\left[{\widehat{Q}_{n}(\theta_{1},\gamma)}-\lambda \|\gamma\|_{1}\right].\label{test-stat-def-1}
\end{align}
Partition $G\left(X_i, \theta\right)=\left[ G_{1}\left(X_i, \theta\right)',G_{2}\left(X_i, \theta\right)' \right]'$,
corresponding to the partial derivatives with respect to $\theta_{1}$ and $\theta_{2}$.
Suppressing the dependence on $ \theta_{1}$ in the notation when $\widehat{T}_n$ or $\zeta_{ni}$ is evaluated at $ \theta_{1,0}$, we obtain the following result.

\begin{thm}
\label{thm:subvector}
Let Assumptions \ref{assum0} and \ref{ass:expansion} hold. Then, 
\[
\widehat{T}_{n}(\lambda)\cvw\sup_{\gamma \in \Gamma}
\left[
\frac{\left|\bar{\mathcal{M}}\left(\gamma\right)+Z' \mathrm{cov} \left[ G_{2}\left(X_i,\theta_{0}\right),\exp\left(W_i'\gamma\right) \right] \right|}{s\left(\gamma\right)}-\lambda\left\Vert \gamma\right\Vert _{1}
\right] \ \text{ in } L_\infty(\Lambda),
\]
where $(Z, \bar{\mathcal{M}} (\gamma))$ is a centered Gaussian random vector with
$\mathbb{E} \left[ ZZ' \right]=\lim_{n\to\infty}
\mathbb{E} \left[ \zeta_{ni} \zeta_{ni}'\right]$ 
and $\mathbb{E} \left[ Z\bar{\mathcal{M}}\left(\gamma\right) \right] =\lim_{n\to\infty}
\mathbb{E} \left[ U_i \zeta_{ni}(\exp\left(W_i'\gamma\right)-\mathbb{E}\exp\left(W_i'\gamma\right)) \right]$ for each $\gamma$ and $\bar{\mathcal{M}} (\gamma)$ is a centered Gaussian process such that $\mathbb{E}\bar{\mathcal{M}} (\gamma_1)\bar{\mathcal{M}} (\gamma_2) = \mathbb{E}U_i^2 (\exp\left(W_i'\gamma_1\right)-\mathbb{E}\exp\left(W_i'\gamma_1\right))(\exp\left(W_i'\gamma_2\right)-\mathbb{E}\exp\left(W_i'\gamma_2\right))$.
\end{thm}

In the presence of pre-estimates in the test statistic, the
multiplier bootstrap in \eqref{mult-bts} is not valid.
To develop valid inference, we now describe how to modify the multiplier bootstrap by
exploiting the influence function $\zeta_{ni}$.
To ease notation, define
\begin{align*}
\widehat{G}_{2i} &:= G_{2}\left(X_{i},\left(\theta_{1,0},\widehat{\theta}_{2}\right)\right)\\
\overline{W}_{ni}(\gamma) &:= \exp(W_{i}'\gamma) - \frac{1}{n}\sum_{j=1}^n \exp(W_{j}'\gamma)
\end{align*}
Let 
\begin{align}\label{modified:1}
\begin{split}
\widehat{M}_{n,*}(\gamma) & :=\frac{1}{n}\sum_{i=1}^{n}\eta_{i}^{*}
\left\{
\widehat{U}_{i}  \overline{W}_{ni} (\gamma)
+\widehat{\zeta}_{ni}^{\, '} \ion\sum_{j=1}^{n} \widehat{G}_{2j} \overline{W}_{nj}(\gamma)
\right\},  \\
\widehat{s}_{n,*}^{2}(\gamma) & :=\frac{1}{n}\sum_{i=1}^{n}\left[
\eta_{i}^{*}
\left\{
\widehat{U}_{i}\overline{W}_{ni}(\gamma) 
+\widehat{\zeta}_{ni}^{\, '} \ion\sum_{j=1}^{n} \widehat{G}_{2j} \overline{W}_{nj}(\gamma)
\right\}\right]^{2},
\end{split}
\end{align}
where $\widehat{\zeta}_{ni}$ denotes a consistent estimator of $\zeta_{ni}$.
Then, we proceed with these modified quantities, as in Section \ref{sec:bootstrap}.
That is, to implement the bootstrap, we need to obtain an explicit plug-in formula for the influence function $\zeta_{ni}(\theta_1)$ in Assumption \ref{ass:expansion}, similar to $g_1(X_i,\theta_1) - \frac{1}{n} \sum_{i=1}^n g_1(X_i, \theta_1)$ from \eqref{eq:theta2hat}.\footnote{Instead, 
one could implement alternative bootstrap without using the explicit formula of the influence function as in \cite{Chen:Linton:VanKeilegom}. We opted not to consider this in this paper.}

\subsection{Choice of Penalty}\label{sec:choice-penalty-pre-estimated}

We start with a sequence of local alternatives $\theta_{1n}=\theta_{1,0}+B/\sqrt{n}$.
Then, expressing the hypothesized value of $\theta_{1n}$ explicitly,
we write the corresponding statistics by
\begin{align*}
\widehat{M}_{n}(\theta_{1n},\gamma) & :=\frac{1}{n}\sum_{i=1}^{n}\widehat{U}_{i}\left(\theta_{1n}\right)\overline{W}_{ni}(\gamma),\\
\widehat{s}_{n}^{2}(\theta_{1n},\gamma) & :=\frac{1}{n}\sum_{i=1}^{n}\left[\widehat{U}_{i}\left(\theta_{1n}\right)\overline{W}_{ni}(\gamma)\right]^{2},\\
\widehat{Q}_{n}(\theta_{1n},\gamma) & :=\sqrt{n}\frac{|\widehat{M}_{n}(\theta_{1n},\gamma)|}{\widehat{s}_{n}(\theta_{1n},\gamma)},
\end{align*}
and the test statistic
\begin{align}
\widehat{T}_{n}\left(\theta_{1n}\right) & :=\sup_{\gamma \in \Gamma}\left[{\widehat{Q}_{n}(\theta_{1n},\gamma)}-\lambda_{n}\|\gamma\|_{1}\right].\label{test-stat-def-1-2}
\end{align}
Define $g_i(\theta_1, \theta_2) := g[X_i, \{\theta_1, \theta_2\}]$ and partition $G(X_i, \{\theta_1,\theta_2\})  = [G_{1i}(\theta_1, \theta_2)', G_{2i}(\theta_1, \theta_2)']'$ as before.
The limit of the test statistic $\widehat{T}_{n}\left(\theta_{1n}\right)$
can be easily obtained by modifying the proof of Theorem \ref{thm:subvector}.
Specifically, we note that under the additional assumption that $\zeta_{ni}(\theta_1)$ is differentiable with respect to $\theta_1$,
\begin{align*}
	\sqrt{n}\widehat{M}_{n}(\theta_{1n},\gamma) & =\frac{1}{\sqrt{n}}\sum_{i=1}^{n}g_{i}\left(\theta_{1n},\theta_{2}(\theta_{1n})\right)\overline{W}_{ni}(\gamma)\\
	& \quad+\frac{1}{\sqrt{n}}\sum_{i=1}^{n}\zeta_{ni}(\theta_{1n})'\frac{1}{n}\sum_{j=1}^{n}G_{2j}\left(\theta_{1n},\theta_{2}(\theta_{1n})\right)\overline{W}_{nj}(\gamma)+o_{p}\left(1\right)\\
	& =\frac{1}{\sqrt{n}}\sum_{i=1}^{n}g_{i}\left(\theta_{0}\right)\overline{W}_{ni}(\gamma) 
 + \frac{1}{\sqrt{n}}\sum_{i=1}^{n}\zeta_{ni}'\frac{1}{n}\sum_{j=1}^{n}G_{2j}\left(\theta_{0}\right)\overline{W}_{nj}(\gamma) \\
&+B'\frac{1}{n}\sum_{i=1}^{n}\left[ G_{1i}\left(\theta_{0}\right)+ 
\frac{\partial\theta_{2}\left(\theta_{1,0}\right)} {\partial\theta_{1}} '
       G_{2i}\left(\theta_{0}\right)\right] \overline{W}_{ni}(\gamma)
	+o_{p}\left(1\right),
\end{align*}
using the fact that 
$n^{-1} \sum_{i=1}^n \frac{\partial\zeta_{ni}\left(\theta_{1,0}\right)} {\partial\theta_{1}} = o_p(1)$.
Thus, the noncentrality term is determined by the probability limit of $B' \mathbb{E} \left[ \omega_{i}(\gamma) \right]$, where 
\begin{align*}
\omega_{i}(\gamma) := \left[ G_{1i} \left(\theta_{0}\right) +  \frac{\partial\theta_{2}\left(\theta_{1,0}\right)} {\partial\theta_{1}}' G_{2i}\left(\theta_{0}\right)\right] \overline{W}_{ni}(\gamma).
\end{align*}
As shorthand notation, let $G_{1i} := G_{1}\left(X_i,\theta_{0}\right)
$ and $G_{2i} :=G_{2}\left(X_i,\theta_{0}\right)$.
We now adjust \eqref{eq:calT} in Section \ref{sec:lambda:calibration}
as follows: let
\begin{align}\label{eq:calT-1}
\begin{split}
&\mathcal{T}\left(\lambda,B\right) \\
&=1\left\{ \sup_{\gamma \in \Gamma}
\left|\frac{\bar{\mathcal{M}}\left(\gamma\right)+Z' \mathrm{cov} \left[ G_{2i},
\exp\left(W_i'\gamma\right) \right] +B' \mathbb{E} [ \omega_i (\gamma) ]}{s\left(\gamma\right)}\right|-\lambda\left\Vert \gamma\right\Vert _{1}>c_{\alpha} (\lambda) \right\}
\end{split}
\end{align}
for a critical value $c_{\alpha}(\lambda)$ and
$\mathcal{R}\left(\lambda,B\right)=\mathbb{E} [ \mathcal{T}\left(\lambda,B\right) ]$.
Then, as before, choose $\lambda$ by solving \eqref{max-min}.
To implement this procedure, we modify the steps in  Section \ref{sec:lambda:calibration}
with 
\begin{align*}
\widehat{M}_{n,*,B}(\gamma) & :=\frac{1}{n}\sum_{i=1}^{n} \left[ \eta_{i}^{*}
\left\{
\widehat{U}_{i}\overline{W}_{ni}(\gamma)
+\widehat{\zeta}_{ni}^{\,'} \ion\sum_{j=1}^{n} \widehat{G}_{2j} \overline{W}_{nj}(\gamma)
\right\} + \frac{B}{\sqrt{n}}\widehat{\omega}_i(\gamma) \right],\\
\widehat{s}_{n,*,B}^{2}(\gamma) & :=\frac{1}{n}\sum_{i=1}^{n}\left[
\eta_{i}^{*}
\left\{
\widehat{U}_{i}\overline{W}_{ni}(\gamma)
+\widehat{\zeta}_{ni}^{\,'} \ion\sum_{j=1}^{n} \widehat{G}_{2j} \overline{W}_{nj}(\gamma)
\right\} + 
\frac{B}{\sqrt{n}} \widehat{\omega}_i(\gamma) 
\right]^{2},
\end{align*}
where
\begin{align*}
\widehat{\omega}_i(\gamma) := \left( G_{1i} + 
 \frac{\partial\widehat\theta_{2}\left(\theta_{1,0}\right)} {\partial\theta_{1}}'
 \widehat{G}_{2i} \right) \overline{W}_{ni}(\gamma).
\end{align*}
Then, the remaining steps are identical to those in Section \ref{sec:lambda:calibration}.

\section{Inferring the Elasticity of Intertemporal Substitution}\label{sec:emp2}

In this section, we revisit \citet{yogo2004} and conduct inference on the elasticity of intertemporal substitution (EIS).
We investigate the case of the annual US series (1891--1995) used in \citet{yogo2004}, focusing on $U_t (\theta) = \Delta c_{t}  - \theta_2 - \theta_1 r_t$, where $\Delta c_{t}$ represents the consumption growth in year $t$ and $r_t$ denotes the real interest rate.
The parameter of interest is EIS, denoted by $\theta_1$.
The instruments $W_t$ consist of the two-period lags of the nominal interest rate, inflation, consumption growth, and the log dividend-price ratio. 
Before applying our method, we studentized each instrument and then applied the transformation $\tan^{-1}(\cdot)$.
The transformed instruments are denoted by $\widetilde{W}_t$.
This ensures that each component of $\widetilde{W}_t$ is bounded and comparable in scale.
The data consist of $\{ (\Delta c_{t}, r_t, W_t): t=1,\ldots,n \}$, where the time span is $n = 105$.

To perform inference on $\theta_1$ in the presence of $\theta_2$, we use the demeaned version of the generalized residuals, defined as follows:
\begin{align*}
\widehat{U}_t (\theta_1) = \left( \Delta c_{t} - \frac{1}{n} \sum_{t=1}^n \Delta c_{t}\right)  - \theta_1 \left( r_t - \frac{1}{n}  \sum_{t=1}^n r_t \right).
\end{align*}
Following the notation in Section~\ref{sec:subvector}, we have in this example,
\begin{align*}
\theta_{2}(\theta_1) &= \mathbb{E} \left[ \Delta c_{t} \right]   - \theta_{1} \mathbb{E} \left[ r_t \right], \\
\widehat\theta_{2}(\theta_1) &= \frac{1}{n} \sum_{t=1}^n \Delta c_{t}   - \theta_{1} \frac{1}{n}  \sum_{t=1}^n r_t, \\
\eta_{nt}(\theta_1) &= \left( \Delta c_{t} -  \mathbb{E}[\Delta c_{t}] \right)  - \theta_1 \left( r_t - \mathbb{E}[ r_t ] \right), \\
\widehat\eta_{nt}(\theta_1) &= \widehat{U}_t (\theta_1).
\end{align*}
Then, because $G_{2t} =  -1$, adopting \eqref{modified:1}  yields
 the following multiplier bootstrap:
 \begin{align}\label{mult-bts-yogo}
\begin{split}
\widehat{M}_{n,\ast}(\gamma) & =\frac{1}{n}\sum_{t=1}^{n}\eta_{t}^{\ast}\widehat{U}_{t}(\theta_1)\left\{ \exp(\widetilde{W}_t'\gamma) - \frac{1}{n} \sum_{t=1}^n  \exp(\widetilde{W}_t'\gamma) \right\},\\
\widehat{s}_{n,\ast}^{2}(\gamma) & =\frac{1}{n}\sum_{t=1}^{n}\left[\eta_{t}^{\ast} \widehat{U}_{t}(\theta_1)\left\{ \exp(\widetilde{W}_t'\gamma) - \frac{1}{n} \sum_{t=1}^n  \exp(\widetilde{W}_t'\gamma) \right\}\right]^{2}.
\end{split}
\end{align}
Furthermore, since $G_{1t} = -r_t$, we use the following to calibrate the optimal $\lambda$:
\begin{align*}
\widehat{M}_{n,*,B}(\gamma) & =\frac{1}{n}\sum_{t=1}^{n}
 \left\{ \eta_t^{\ast} \widehat U_t(\theta_1) - \frac{B }{\sqrt n}(r_t - \bar r)\right\} \left\{ \exp(\widetilde{W}_t' \gamma) - \frac 1 n \sum_{t = 1}^ n \exp(\widetilde{W}_t' \gamma)\right\}
,\\
\widehat{s}_{n,*,B}^{2}(\gamma) & =\frac{1}{n}\sum_{t=1}^{n}
\left[ \left\{ \eta_t^{\ast} \widehat U_t(\theta_1) - \frac{B }{\sqrt n}(r_t - \bar r)\right\}\left\{ \exp(\widetilde{W}_t' \gamma) - \frac 1 n \sum_{t = 1}^ n \exp(\widetilde{W}_t' \gamma)\right\} \right]^2,
\end{align*}
where $\overline{r} = n^{-1} \sum_{t=1}^n r_t$.



\paragraph{Computation of the Test Statistic}

In Table~\ref{tab:true-max}, we report the values of the max statistic
$$
T(\lambda, a) = \sup_{\gamma \in [-a,a]^4} \left[ \widehat Q_n(\bar{\theta}_{1},\gamma) - \lambda \|\gamma\|_1\right],
$$
evaluated at $\bar{\theta}_{1} = -0.028$, while varying the penalty level $\lambda \in \{ .3, .2, .1, 0 \}$ and the domain constant $a \in \{ 1,2,3,4,5 \}$.
This value of $\bar{\theta}_{1}$ is chosen as the 2SLS estimate from the data.
To ensure the accuracy of the algorithm, the max statistic is computed with a swarm size of 5000, which spans the search space $[-a, a]^4$ uniformly at random.

\paragraph{Computation Time}

The swarm size of 5000, used to compute $T(\lambda,a)$ across $(\lambda,a)$-values in Table~\ref{tab:true-max}, enables highly accurate optimization. However, it may be inappropriate for practical applications.
Therefore, we examine the impact of $\lambda$ and $a$ on computation time based on a more realistic swarm size of 200.
In light of the inherent randomness in swarm generation, we regenerate the swarm randomly and recompute the max statistic until it achieves at least 95\% of the maximum values reported in Table~\ref{tab:true-max}.
That is, we measure the computation times under the constraint that the optimization error does not exceed 5\%.
In Table~\ref{tab:computation-times}, we report the means and standard deviations of the computation times across various $(\lambda, a)$-values.

Tables~\ref{tab:true-max} and~\ref{tab:computation-times} reveal three important observations.
First, the average computation time tends to increase as $\lambda$ decreases across all $a$ values.
This aligns with the intuition that introducing the penalty shrinks the feasible space for $\gamma$, thereby leading to shorter search times for the maximum.
Second, the computation times are numerically most stable when $\lambda = 0.3$ is applied.
The impact of $a$ on numerical instability tends to increase with $a$ and is greatest when $\lambda$ is set to $0$.
Lastly, the value of the unpenalized statistic is more sensitive to the choice of $a$ than it is with penalization.
This suggests the importance of implicit search space selection in Bierens' max test with $\lambda = 0$, as it may significantly impact the outcomes of the analysis.
Our method is motivated by tuning a scalar parameter $\lambda$ rather than $\Gamma = [-a,a]^4$, which makes it suitable to optimize in a data-driven manner.

\paragraph{Calibration of Optimal Penalty}

The optimal $\lambda$ in $\Lambda = \{0.5, 0.4, 0.3, 0.2, 0.1, 0.05, 0\}$ is calibrated as in \eqref{max-min}, based on a single local alternative $\mathcal{B} = \{ 2 \}$.
The local alternative can be set to a singleton since $\eta_{t}^{\ast} \sim N(0,1)$ is symmetrically distributed about zero, leading to an increase in power with the absolute value of $B$.
We set $\alpha$ to $0.1$ and use the grid $\Theta_1 = \{-0.6, -0.4, -0.2, \ldots, 0.6 \}$ for hypothesized values of $\theta_{1,0}$.
For each $\lambda \in \Lambda$ and $\theta_{1,0} \in \Theta_1$, the local power in \eqref{RP:bootstrap} is simulated following the steps outlined in Section~\ref{sec:choice-penalty-pre-estimated}.

We use two algorithms for computing the max statistic to assess the sensitivity to the choice of the algorithm.
The PSO algorithm searches for the global maximum in the entire space $\Gamma = [-5, 5]^4$. 
In each computation, a swarm with a size of 200 is generated uniformly at random over $\Gamma$.
The grid search (GS) method calculates the maximum over a discretized grid $\Gamma' \subset [-5,5]^4$, chosen to minimize potential power loss while leveraging a vectorized algorithm to avoid loops used in the PSO algorithm and other optimization methods.
The tuning parameter in the grid search corresponds to the choice of the grid $\Gamma'$.
We opted for the equi-spaced grid $\Gamma' = \{ 0.5 j : -10\le j \le 10 \}^4 = \{ -5,-4.5,\ldots, 5 \}^4$, which consists of $21^4 = 194481$ points. The numbers of bootstrap replications are $R_{\mathrm{PSO}} = 1000$ and $R_{\mathrm{GS}} = 5000$ for the PSO algorithm and the GS method, respectively.

Figure~\ref{fig:local-powers-average-and-pointwise} shows the local powers simulated using these algorithms.
Panel~A displays the average local powers over $\Theta_1$ across different values of $\lambda$, while Panel~B exhibits the local powers at a specific $\theta_{1,0} = 0$.
From Panel~A, we note that the PSO algorithm yields overall higher powers than the GS method, which operates on the discrete $\Gamma'$.
In Panel~B, for both algorithms, we observe a rise in the local power within a small range around $\lambda = 0$, followed by leveling off at approximately $\lambda = 0.2$ or $\lambda = 0.3$.
Since this pattern also prevails for the other values in $\Theta_1$, we choose the optimal $\lambda$ as the maximizer of the average local powers in Panel~A.
The selected $\lambda$ are 0.3 for the PSO algorithm and 0.2 for the GS method, respectively.
These values are used for the optimally penalized test of $H_0:\theta_1 = \theta_{1,0}$ uniformly across the hypothesized values.

Table~\ref{tab:confidence-intervals} presents the confidence intervals based on our optimally penalized and unpenalized tests using each algorithm.
These confidence intervals are constructed by inverting the testing of $H_0: \theta_{1} = \theta_{1,0}$, where the hypothesized values of $\theta_{1,0}$ range over $\{ 0.01 j :-100\le j\le 100 \}$. 
It shows that the optimally penalized test yields significantly narrower confidence intervals than the unpenalized test, irrespective of the chosen algorithm.

\citet{yogo2004} commented that ``there appears to be identification failure for the annual U.S. series.''
Indeed, the 95\% confidence interval from the Anderson-Rubin (AR) test was $[-0.49,0.46]$
and those from the Lagrange multiplier (LM) test and the conditional likelihood ratio (LR) tests were $[-\infty,\infty]$ \citep[see Table 3 of][]{yogo2004}.\footnote{If each instrument is used separately for the AR test, the resulting 95\% confidence intervals are as follows: 
(i) $[-\infty, \infty]$ with the nominal interest rate as an instrument;
(ii) $[-0.29, 0.28]$ with inflation;
(iii) $[-\infty, \infty]$ with consumption growth;
(iv) $[-\infty, -0.12] \cup [0.02, \infty]$ with log dividend-price ratio,
where the instruments are twice lagged in all cases.
The confidence interval using inflation is similar to ours, but its length is more than 25\% larger than our optimal confidence interval [$-0.30, 0.15$] using the PSO algorithm and more than 15\% larger than the optimal confidence interval using the GS method.}
Our penalized test provides tighter confidence intervals than any of these similar tests based on unconditional moment restrictions, suggesting that conditional moment restrictions can be more informative than an arbitrarily selected set of unconditional moment restrictions.

In a nutshell, we demonstrate that a seemingly uninformative set of instruments can provide an informative inference result if one strengthens unconditional moment restrictions by making them infinite-dimensional conditional moment restrictions with the aid of penalization.


\section{Monte Carlo Experiments}
\label{sec:mc}
To examine the efficacy of our method in an empirically relevant context, we conduct a series of experiments based on the annual U.S. series (1891-1995) used in \citet{yogo2004}.
The main purpose of the experiments is to assess the finite-sample size and power properties of our proposed method in comparison to those of existing methods.

The simulated series is denoted by $\{ (\Delta c_t^*, r_t^*, W_t^*) \}_{t=1}^n$ with a $*$-superscript, while the original series is denoted by $\{ (\Delta c_t, r_t, W_t) \}_{t=1}^n$, as in Section~\ref{sec:emp2}.
We consider the following data-generating process in our experiments:
\begin{align}
\label{eq:monte-carlo-DGP}
    r_{t}^* &= (
        1 \ \ \ \ W_{t}^*{}'
) \pi_0 + \bar{\pi} f(W_t^*)  +   v_t^*, \nonumber\\     
    \Delta c_t^* &= (
        1 \ \ \ \ r_{t}^*
    )\theta_0 + u_t^*,\\
    \left( \begin{matrix}
    u_t^* \\ v_t^*    
\end{matrix} \right) & \sim N \left[ \left( \begin{matrix}
    0 \\ 0    
\end{matrix} \right), \left( \begin{matrix}
    \sigma_u^2 & \sigma_u \sigma_v \rho \\
    \sigma_u \sigma_v \rho & \sigma_v^2
\end{matrix} \right) \right]. \nonumber
\end{align}
The linear coefficients $\theta_0$ and $\pi_0$ are taken from their respective estimates in the original series.
Specifically, $\pi_0 = (\pi_{0,0}, \ldots, \pi_{4,0})'$ is computed from the first-stage OLS regression
\begin{equation*}
    r_t = (1 \ \ \ \ W_{t}')  \pi_0 + \hat v_t,    
\end{equation*}
and $\theta_0 = (\theta_{0,0}, \theta_{1,0})'$ is obtained from the 2SLS regression
\begin{equation*}
    \Delta c_t = (1 \ \ \ \ r_{t}) \theta_0 + \hat u_t,    
\end{equation*}
where $W_t$ are used as instruments.

We specify the nonlinear term $f(W_t^*) = \delta_{0,0} + \delta_{1,0}'W_t^* + \delta_{2,0}' (W_{t}^* \odot W_{t}^*) $ as a quadratic function in $W_t^*$ without interaction terms, where $\odot$ denotes the elementwise (Hadamard) product.
The coefficients $\delta_0 = (\delta_{0,0}, \delta_{1,0}', \delta_{2,0}')'$ are determined by the OLS regression of $\hat {v}_t$ on $(1,\ W_t, \ W_t\odot W_t)$, where $\hat{v}_t$ denote the fitted residuals from the first-stage regression.
This choice of $\delta_0$ is intended to replicate the orthogonality between $W_t$ and $f(W_t)$ in the simulated series.
The identification of $\theta_{0}$ based on the given conditional moment restrictions becomes stronger as $|\bar{\pi}|$ increases. 
However, in linear IV models, the identification strength should remain largely unaffected.\footnote{The comparison between conditional and unconditional moments is not new and is most recently discussed at length in \citet{Antoine:Lavergne}.}

The disturbance terms $(u_t^*, v_t^*)$ follow a bivariate normal distribution, drawn independently across periods.
The homoscedastic error variances, $\sigma_u^2$ and $\sigma_v^2$, are taken from the corresponding sample variances of $\hat u_t$ and $\hat v_t$, respectively.
Their correlation coefficient is set to $\rho = 0.8$, which determines the degree of endogeneity.

The data are generated recursively over periods.
Let $(\Delta c_{-1}^*, \Delta c_{0}^*)$ be the initial values, set to the observed values in the data and held constant across simulated series.
In the $t$-th iteration step, where $t \in \{ 1,2,\ldots,n \}$, we first define $W_t^* = (\Delta c_{t-2}^*, W_{-1,t})$, where $W_{-1,t}$ comprises all components of $W_t$ except the twice-lagged consumption growth.
Next, we generate $r_t^*$ and $\Delta c_t^*$ sequentially according to the DGP specified in \eqref{eq:monte-carlo-DGP}.
This procedure is iterated until reaching the final period.

We focus on constructing a confidence set for $\theta_{1,0}$ based on the conditional moment restriction
\begin{equation}
\label{eq:monte-carlo-conditional-moment-model}
    \mathbb{E}[\left. \Delta c_t^* - \mathbb{E}[\Delta c_t^*] - \theta_{1,0} (r_t^* -\mathbb{E}[r_t^*]) \right| W_t^*] = 0.
\end{equation}
The strength of identification can be adjusted by varying the constant $\bar{\pi}$ associated with the nonlinear term.
In the baseline specification, we set $\bar{\pi} = 2$ to make the variation in the linear and nonlinear components comparable to each other.

We follow the same procedure as outlined in Section~\ref{sec:emp2}, which involves demeaning both residuals and exponential weights.
To construct the test statistic, we calculate the demeaned generalized residuals as $\widehat U_t(\theta_1) = \Delta c_t^* - \bar \Delta c^* - \theta_1 (r_{t}^* - \bar r^*)$, where $\bar \Delta c^* = \frac 1 n \sum_{t=1}^n \Delta c_t^*$ and $\bar r^* = \frac 1 n \sum_{t = 1}^ n r_{t}^*$.
The exponential weights are computed based on $\widetilde{W}_t^*$, defined analogously to $\widetilde{W}_t$.
A similar demeaning is also applied to the exponential weights.

The test statistic is computed as defined in \eqref{test-stat-def-1}.
To compute this maximum, we use the grid search (GS) method, adopting the same grid $\Gamma'= \{ -5,-4.5, -4,\ldots, 5 \}^4$ over $[-5,5]^4$ as in Section~\ref{sec:emp2}.
Despite concerns regarding potential power loss, we opted for the GS method over the PSO algorithm for simulations.
This decision was based on its computational efficiency gained by using vectorized code.
Some preliminary experiments with more refined grids support that $\Gamma'$ successfully generates sufficient powers.

The optimal penalty level is selected from $\Lambda = \{0.1 j : j=0,\ldots, 10 \} \subset [0, 1]$. To compute the optimal $\lambda$, we employ the following multiplier bootstrap:
\begin{align*}
    \widehat M_{n, *,B}(\gamma) &= \frac 1 n \sum_{t = 1}^ n \left\{ \eta_t^* \widehat U_t(\theta_{1}) - \frac{B }{\sqrt n}(r_t^* - \bar r^*)\right\} \left\{ \exp(W_t^*{}' \gamma) - \frac 1 n \sum_{t = 1}^ n \exp(W_t^*{}' \gamma)\right\}, \\
    \widehat s_{n, *,B}^2(\gamma) &= \frac 1 n \sum_{t = 1}^ n \left[ \left\{ \eta_t^* \widehat U_t(\theta_{1}) - \frac{B }{\sqrt n}(r_t^* - \bar r^*)\right\}\left\{ \exp(W_t^*{}' \gamma) - \frac 1 n \sum_{t = 1}^ n \exp(W_t^*{}' \gamma)\right\} \right]^2,
\end{align*}
where $B/\sqrt n = 2/\sqrt n$ approximately corresponds to $0.2$ on the actual scale of $\theta_1$, and $\eta_t^*$ are drawn independently from $N(0, 1)$.
The optimal penalty is then calibrated for each $\theta_1 \in \Theta_1 := \{ \theta_{1,0} + 0.1 j : |j|\le 6\}$ following the procedure specified in Section~\ref{sec:choice-penalty-pre-estimated}. The number of bootstrap replications is set to $R = 5000$, mirroring the empirical application. Summary statistics regarding the distribution of optimal $\lambda$ are presented in Table~\ref{tab:selection-optim-penalty} across various values of $\theta_1$ in $\Theta_1$.

We assess the size and power properties of our optimal penalized and unpenalized tests in comparison to those of the AR test and the Wald tests based on the 2SLS estimator and that of \citet{dominguez2004consistent}.
Figure~\ref{fig:power-curves-4IVs-Yogo-MC} displays the power curves depicting the performance of each method.
This indicates that only the AR and our tests maintain a size close to the nominal level of 0.1.
In contrast, the Wald tests based on the estimators of $\theta_1$ exhibit significant size distortion.
Notably, our optimal test not only achieves nearly accurate size control but also enhances the power compared to both the unpenalized test (with or without size adjustments) and the AR test.\footnote{There is seemingly no formal guidance on how to adjust the undersize of the unpenalized test to fairly compare its power with the other tests. For visual presentation, we depict `Unpen. w/ shift' by shifting the power curve up parallel to match the nominal size of 0.1 at $\theta_1 - \theta_{1,0} = 0$.}
This underscores our motivation that $\ell_1$-regularization can enhance the power of the test by selecting relevant information from the conditioning variables.

We conducted additional experiments where we included three- and four-period lags of the consumption growth rate as supplementary instruments. 
This amounts to adding noise to the IVs in our setting.
Regardless of the value of $\bar{\pi}$, it holds that 
\begin{equation*}
    \mathbb{E}^*\left[\left. \frac{\partial U_t(\theta_{1,0})}{\partial \theta_1} \right| W_t^*, \Delta c_{t-3}^*, \Delta c_{t-4}^* \right] = \mathbb{E}^*\left[\left. \frac{\partial U_t(\theta_{1,0})}{\partial \theta_1}\right| W_t^* \right].
\end{equation*}
Since $\Delta c_{t-3}^*$ and $\Delta c_{t-4}^*$ fail to provide additional explanatory power, it is expected that incorporating these variables as instruments would lead to reduced statistical power, compared to using the original set of 4 IVs.
Table~\ref{tab:power-6-IVs-quadratic-spec} confirms this insight; however, it also reveals that the degree of power reduction is much less severe for the optimally penalized test than for the unpenalized test and the AR test.\footnote{As a caveat, it is unlikely but possible that power reduction might have occurred due to the reduced time span from 105 to 103.}
This implies that the relative advantage of optimal penalization increases as the number of uninformative or less informative instruments becomes large.

We also investigated the size and power of each test under an alternative specification of $\bar{\pi}$, where we nullified the nonlinear term by setting $\bar{\pi}$ to $0$ while keeping the other coefficients the same.
Table~\ref{tab:power-6-IVs-linear-spec} presents the results from this analysis.
There are two noteworthy observations.
First, our tests nearly hold the correct size even when IVs are weak both conditionally and unconditionally.
This demonstrates the robustness of our method to the presence of weak IVs.
Second, excluding the nonlinear term leads to diminished powers of the penalized and unpenalized tests relative to the AR test.
This indicates that a substantial portion of the power in our tests stems from an underlying nonlinear relationship between the noncentrality term and IVs.
Furthermore, our results suggest that the nonlinearity term, if present, can serve as a valuable source of identification in IV models.

\section{Conclusions}\label{sec:concl}

We have developed an inference method for a vector of parameters using an
$\ell_1$-penalized  maximum statistic.
Our inference procedure is based on the multiplier bootstrap
and combines inference with model selection  to improve the power of the  test.
We have recommended solving a data-dependent max-min problem to select the penalization tuning parameter.
We have demonstrated the efficacy of our method using two empirical examples.

There are multiple directions to extend our method. First, we may consider a panel data setting where the number of conditioning variables may grow as the time series dimension increases. 
Second, unknown parameters may include an unknown function
\citep[e.g.,][]{chamberlain1992efficiency,Newey:Powell:03,Ai:Chen:03,chen2015inference}.
In view of results in \cite{Breunig:Chen},
Bierens-type tests without penalization might not work well when the parameter of interest is a nonparametric function.
It would be interesting to study whether and to what extent our penalization method improves power for nonparametric inference.
Third, multiple conditional moment restrictions or a continuum of conditional moment restrictions (e.g., conditional independence assumption) might be relevant in some applications.
Fourth, it would be interesting to extend our method for empirical industrial organization. For instance, \citet{Gandhi:Houde:NBER} proposed a set of relevant instruments from conditional moment restrictions to avoid the weak identification problem.  It is an intriguing possibility to combine our approach with their insights into \citet{BLP}.
All of these extensions call for substantial developments in both theory and computation.

\section{Proofs}\label{sec:proofs}

\begin{proof}[Proof of Theorem \ref{thm:fixed p null}]
	First, we show the stochastic equicontinuity of the processes $ \sqrt{n} M_n (\gamma)  $ and $  s_n (\gamma)  $ for \eqref{eq:FCLT} and \eqref{eq:ULLN}.
	Due to the boundedness of $ \Gamma $ and $ W_i $, $ U_i \exp(W_i ' \gamma) $ is Lipschitz continuous  with a bound $ K |U_i| \| \gamma_{1} - \gamma_{2} \|  $ for some $ K $ and for any $ \gamma_1 $ and $ \gamma_{2} $. Then, this Lipschitz property and the existence of moment of some $c > d$ implies due to Theorem 2 in \cite{Hansen:96} that the empirical process $ \sqrt{n} M_n (\gamma) $ is stochastically equicontinuous. 
	The Lipschitz continuity and the ergodic theorem also imply that $ s_n (\gamma) $  is stochastically equicontinuous. 
	
	Next, the martingale difference sequence central limit theorem and the ergodic theorem yield the desired finite-dimensional convergence 	 for \eqref{eq:FCLT} and \eqref{eq:ULLN} under Assumption \ref{assum0}; see e.g. \cite{davidson1994stochastic}'s Section 24.3 and 13.4.
	 
Finally, for the convergence of ${T}_{n}(\lambda)$, note that both $ \Lambda $ and $ \Gamma $ are bounded, implying $ \lambda \|\gamma\|_{1} $ is uniformly continuous. Thus, the process $ \frac{\left|\mathcal{M}\left(\gamma\right)\right|}{s\left(\gamma\right)}-\lambda \|\gamma\|_{1} $ converges weakly in $ \ell^{\infty}(\Gamma \times \Lambda) $, the space of bounded functions on $ \Gamma \times \Lambda $, and the weak convergence of 	
	 $T_n ( \lambda) $ follows from the continuous mapping theorem since (elementwise) $ \sup $ is a continuous operator.
\end{proof}

\begin{proof}[Proof of Theorem \ref{thm:fixed p null:bootstrap}]
For the same reason as in the proof of Theorem \ref{thm:fixed p null}, it is sufficient to verify the conditional finite dimensional convergence. As $ \eta_{i}^{\ast} g(X_i,\bar{\theta }) \exp(W_i'\gamma) $ is a martingale difference sequence, we verify the conditions in \cite{hall:heyde:martingale}'s  Theorem 3.2, a conditional central limit theorem for martingales. 
Their first condition that  
\[  n^{-1/2}\sup_{i}\left\vert \eta_{i}^{\ast} g(X_i,\bar{\theta } ) \exp(W_i'\gamma)\right\vert \overset{p}{\rightarrow } 0  \] and the last condition
$\mathbb{E} \left[ \sup_{i}\eta_{i}^{*2} g(X_i,\bar{\theta } )^2 \exp(2W_i'\gamma)\right] =O(n) $ are straightforward since $ \exp(W_i'\gamma) $ is bounded and $ |\eta_{i}^{\ast} g(X_i,\bar{\theta } )|$
has a finite $ c  $ moment for $ c>2 $.
Next, \[  n^{-1}\sum_{i=1}^{n} \eta_{i}^{*2} g(X_i,\bar{\theta } )^2 \exp(2W_i'\gamma) \overset{p}{\rightarrow} \mathbb{E} \left[ g(X_i,\bar{\theta } )^2 \exp(2W_i'\gamma) \right] \]
by the ergodic theorem.
This completes the proof.
\end{proof}

\begin{proof}[Proof of Theorem \ref{thm:consistency}]
	It follows from Lemma \ref{Lem:Bierens} that
	$$
	\frac{\left| \mathbb{E} \left[ g\left(X_i,\bar{\theta}\right)\exp(W_i'\gamma) \right] \right|}
	{\sqrt{\mathbb{E} \left[ g^2 \left(X_i,\bar{\theta}\right)\exp(2W_i'\gamma) \right] }}>0
	$$
	for almost every $ {\gamma \in \Gamma} $.
	Then, the result follows from the ergodic theorem.	
\end{proof}

\begin{proof}[Proof of Lemma \ref{lem:power gain}] 
Since the derivative of $ T(\lambda,b) $ with respect to  $\lambda$ at $\lambda=0$ is $-\left\Vert \tilde{\gamma}(b) \right\Vert_1 $ 
for each $b=0,B$ , the difference between the alternative and null limit experiments $ T(0,B)-T(0,0) $ at $\lambda=0$ is stochastically dominated by $ T(\lambda,B)-T(\lambda,0) $ at a positive $\lambda$. This implies that the rejection probability of the test at any prespecified significance level is bigger at the experiment with the positive $\lambda$. Thus,  $ T(\lambda, 0)-T(0, 0)  < T(\lambda, B) - T(0,B) \quad \text{a.s.} $
\end{proof}

\begin{proof}[Proof of Theorem~\ref{thm:consistency of lambda}]
	We begin with showing that	$c_{\alpha}^{*} (\lambda) \cvp c_{\alpha} (\lambda)$ uniformly in $ \Lambda $.  First, recall that the inverse map on the space of the distribution function $ F $ that assigns its $ \alpha $-quantile is Hadamard-differentiable at $ F $ provided that $ F $ is differentiable at $ F^{-1}(\alpha) $ with a strictly positive derivative; see e.g. Section 3.9.4.2 in \cite{VW:1996:book}.  Therefore, for the uniform consistency of the bootstrap, it is sufficient to show that $  F^{*}_{\lambda} (x)  \cvp  F_{\lambda} (x) $ uniformly $ (x,\lambda) \in A_{\alpha} $. 
	However, this is a direct consequence of the conditional stochastic equicontinuity and the convergence of the finite-dimensional distributions established in Theorem \ref{thm:fixed p null:bootstrap}.
	
	Next, the preceding step implies that $ T_{n,*,B}(\lambda) - c^*_{\alpha}(\lambda) $ converges weakly to 
	\begin{equation*}
 \sup_{\gamma \in \Gamma} \left\{ \left|\frac{\mathcal{M}\left(\gamma\right)+
\mathbb{E} \left[ \exp(W_i'\gamma)G_i B \right]}{s\left(\gamma\right)}\right|-\lambda\left\Vert \gamma\right\Vert _{1} \right\}
- c_{\alpha}(\lambda),
\end{equation*}
which is the limit in \eqref{eq:calT}. 
	This in turn yields	the uniform convergence of $ \mathcal{R}_n (\lambda,B) $ in probability	to $ \mathcal{R} (\lambda,B) $. Since  $ \mathcal{R} $ is continuous on a compact set, the standard consistency argument results in that $ d(\widehat{\lambda} , \Lambda_0) \cvp 0  $.
	
	For the same reason, $ T_{n}(\lambda) - c^*_{\alpha}(\lambda) $ converges weakly to $ T(\lambda) - c_{\alpha}(\lambda) $ and thus the probability that $ T_{n}(\widehat{\lambda}) \geq c^*_{\alpha}(\widehat{\lambda}) $ converges to $\alpha$ for any sequence $\widehat{ \lambda}$ due to the weak convergence.

 \end{proof}

\begin{proof}[Proof of Theorem \ref{thm:subvector}]
Write $g_{i}\left(\theta\right)$ and $G_{2i}\left(\theta\right)$
for $g\left(X_{i},\theta\right)$ and $G_2\left(X_{i},\theta\right)$,
respectively.
Note that for $ \theta_{1}=\theta_{1,0}$,
\begin{align*}
\sqrt{n}\widehat{M}_{n}(\gamma) & =\frac{1}{\sqrt{n}}\sum_{i=1}^{n}g_{i}\left(\theta_{0}\right)\left(\exp(W_{i}'\gamma) - \frac{1}{n}\sum_{j=1}^n \exp(W_{j}'\gamma)\right)\\
 & \quad+\frac{1}{\sqrt{n}}\sum_{i=1}^{n}\zeta_{ni}
 \frac{1}{n}\sum_{j=1}^{n}G_{2j}\left(\theta_{0}\right)\left(\exp(W_{j}'\gamma) - \frac{1}{n}\sum_{l=1}^n \exp(W_{l}'\gamma)\right)+o_{p}\left(1\right)
\end{align*}
due to Assumption \ref{ass:expansion}. Then, $ \frac{1}{n}\sum_{j=1}^{n}G_{2j}\left(\theta_{0}\right)\exp(W_{j}'\gamma)$ and $  \frac{1}{n}\sum_{l=1}^n \exp(W_{l}'\gamma) $ converge uniformly in probability and
$ \frac{1}{\sqrt{n}}\sum_{i=1}^{n} a_1 g_{i}\left(\theta_{0}\right)\left(\exp(W_{i}'\gamma) - \mathbb{E} \exp(W_{j}'\gamma)\right) + a_2 \zeta_{ni} $ is P-Donkser
for any real $ a_1 $ and $ a_2 $ for the same reasoning as in the proof of Theorem \ref{thm:fixed p null}.
Similarly, the uniform convergence of $ \widehat{s}_n^2 (\gamma)  $ follows since $ g() $ is Lipschitz in $ \theta $ by \eqref{eq:expansion of g}.
\end{proof}

\bibliographystyle{economet}
\bibliography{CLS.bib}

\newpage

%
\newpage
\thispagestyle{empty}

\begin{figure}[htb!]
	\caption{Graphical Representation of Power Improvements via Penalization}
	\label{fig-power}
	\begin{center}
		\includegraphics[scale=0.7]{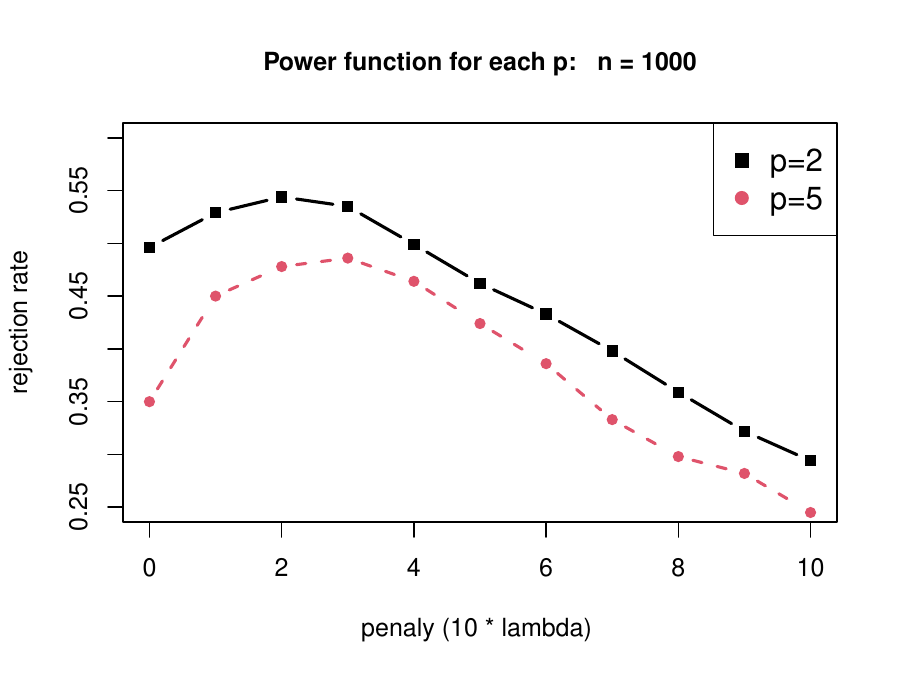}	
	\end{center}
 \begin{center}
\parbox{6in}{Notes: Figure \ref{fig-power} plots the ``theoretical'' power functions of our proposed test, where the power curves are obtained via Monte Carlo simulations described as in the main text with $\Gamma = [-5,5]^p$.}
\end{center}
\end{figure}
\newpage
\thispagestyle{empty}

   \begin{figure}[!htbp]
    	\caption{Simulated Local Powers Computed Using PSO and GS}
    	\label{fig:local-powers-average-and-pointwise}
    	\begin{center}
     	\includegraphics[scale=0.65]{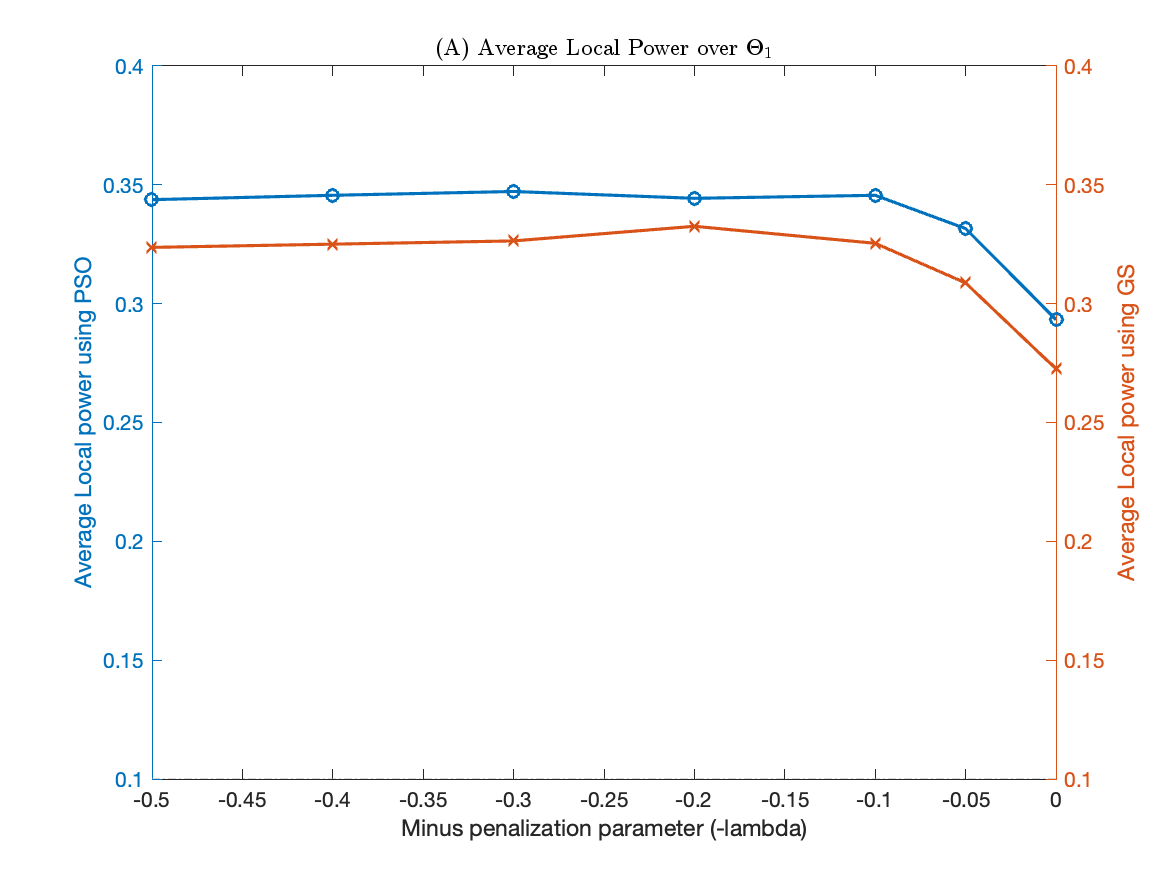}
         \includegraphics[scale=0.65]{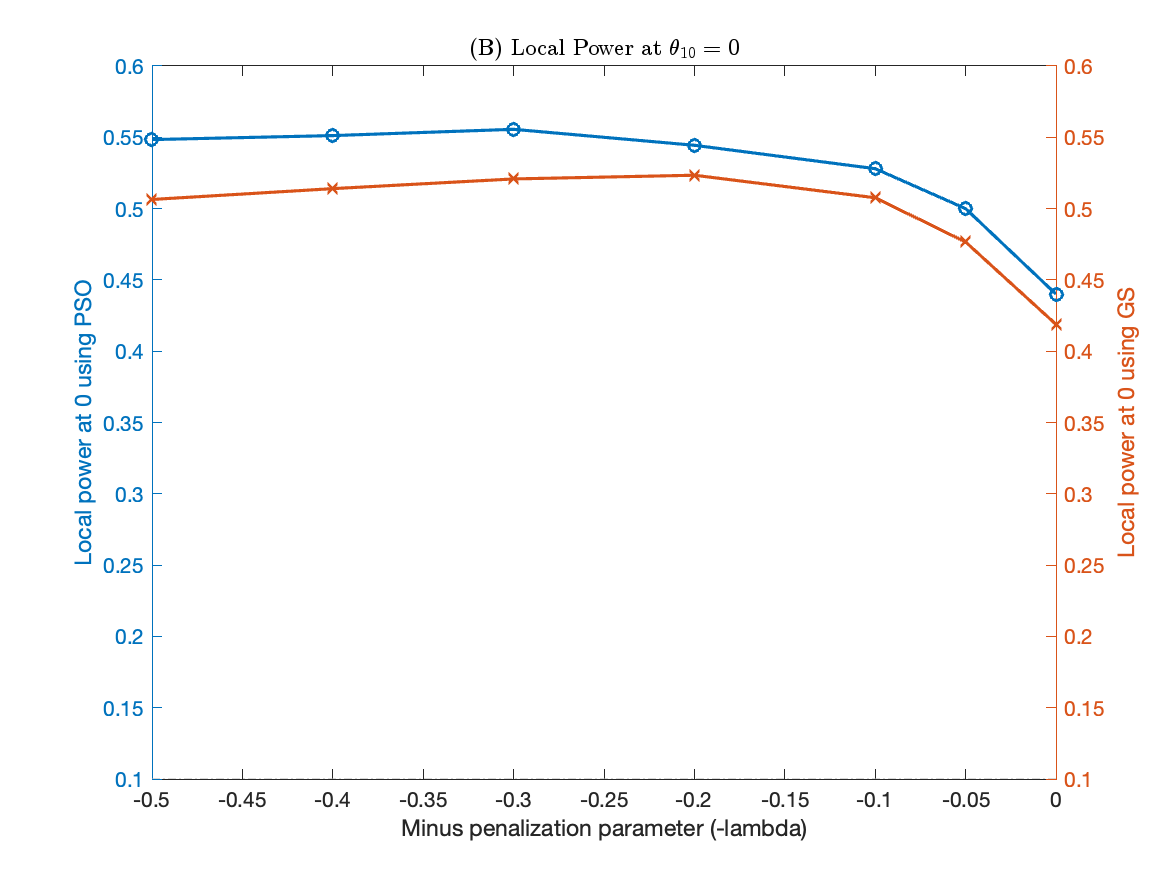}
    	\end{center}
     \begin{center}
     \parbox{6in}{Notes: Figure~\ref{fig:local-powers-average-and-pointwise} shows the simulated local powers using both algorithms, computed as per the eq. \eqref{eq:calT}. Panel~A of Figure~\ref{fig:local-powers-average-and-pointwise} displays the local powers averaged over $\Theta_1$ and Panel~B exhibits the local powers evaluated at a specific $\theta_{1,0} = 0$.
    }	
    \end{center}
    \end{figure}
\newpage
\thispagestyle{empty}

\begin{figure}[!htbp]
	\caption{Power Curves of Competing Tests under the Baseline Specification}
	\label{fig:power-curves-4IVs-Yogo-MC}
	\begin{center}
		\includegraphics[scale=0.65]{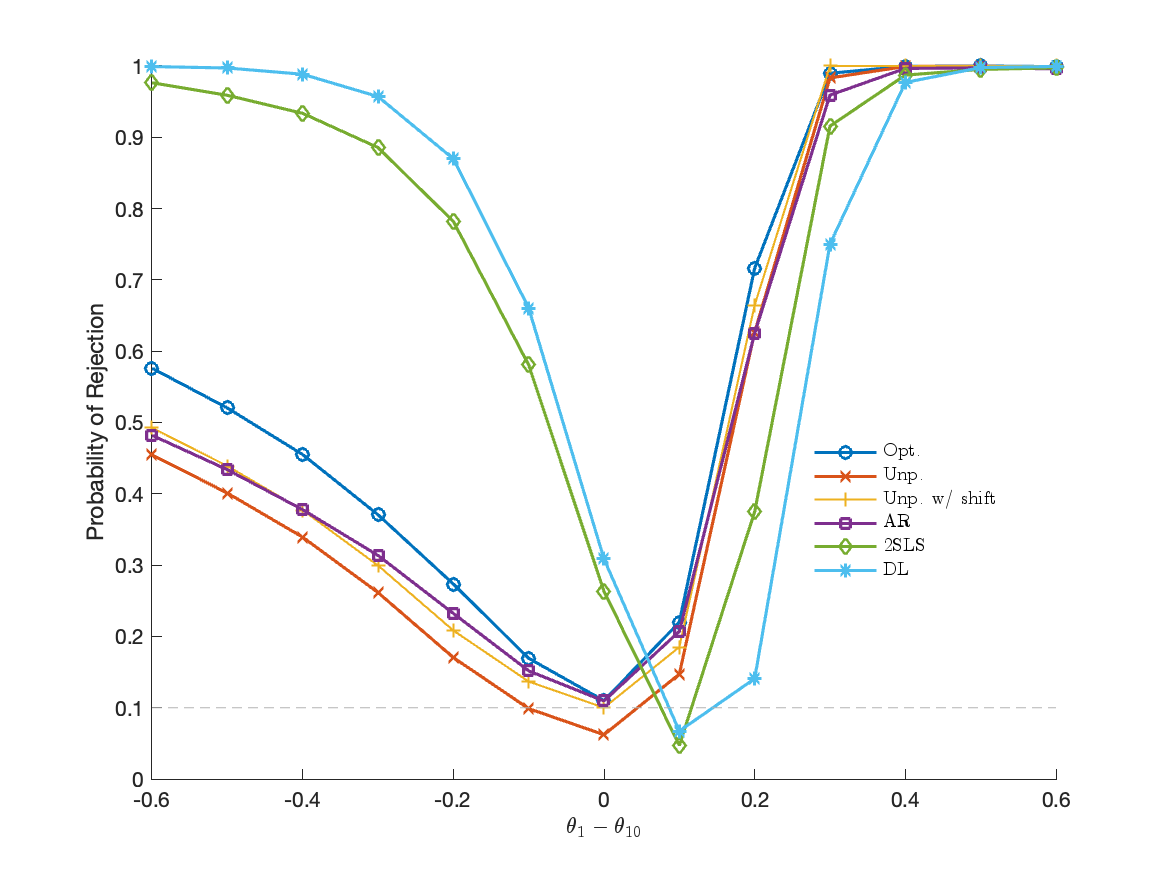}
        \parbox{5in}{\footnotesize{Notes: The nominal size is set to $0.1$. The yellow line, labeled as `Unpen. w/ shift,' is derived by adjusting the size and power of the unpenalized test in parallel to achieve the actual size consistent with the nominal size 0.1. The number of simulation replications used for generating Figure~\ref{fig:power-curves-4IVs-Yogo-MC} is 5000.}
        }
	\end{center}
\end{figure}

\clearpage

\thispagestyle{empty}

\begin{table}[htbp]
\begin{center}
\caption{Values of Max Statistic $T(\lambda, a)$}
\label{tab:true-max}
\begin{tabular}{C{0.075\textwidth} C{0.12\textwidth} C{0.12\textwidth} C{0.12\textwidth} C{0.12\textwidth}}
    \toprule
\multirow[c]{2.5}{*}{$a$}    & \multicolumn{4}{c}{$\lambda$} \\
    \cmidrule(lr){2-5}
    & 0.30 & 0.20 & 0.10 & 0.00 \\ 
    \midrule
    1 & 1.807 & 1.807 & 1.850 &  2.120 \\ 
    2 & 1.807 & 1.807 & 1.850 &  2.176 \\ 
    3 & 1.807 & 1.807 & 1.890 &  2.325 \\ 
    4 & 1.807 & 1.807 & 1.901 &  2.410 \\ 
    5 & 1.807 & 1.807 & 1.901 &  2.448 \\ 
\bottomrule
\end{tabular}
\end{center}
\begin{center}
        \parbox{6in}{\footnotesize{Notes: The values reported in Table~\ref{tab:true-max} are computed with a swarm size of 5000. The value of $\theta_{1,0}$ is fixed at the value of corresponding 2SLS estimate, namely $-0.028$.}}
    \end{center}
\end{table}
%
\newpage
\thispagestyle{empty}
\begin{table}[htbp]
\begin{center}
\caption{Mean and Standard Deviation of Computation Time for Max Statistic}
\label{tab:computation-times}
\begin{tabular}{C{0.075\textwidth} C{0.12\textwidth} C{0.12\textwidth} C{0.12\textwidth} C{0.12\textwidth}}
    \toprule
\multirow[c]{2.5}{*}{$a$}    & \multicolumn{4}{c}{$\lambda$} \\
    \cmidrule(lr){2-5}
    & 0.30 & 0.20 & 0.10 & 0.00 \\ 
    \midrule
    \multirow[c]{2}{*}{$1$} & 3.440 & 3.408 & 3.466 & 3.584 \\ 
 & \footnotesize{(0.070)} & \footnotesize{(0.097)} & \footnotesize{(0.114)} & \footnotesize{(0.110)} \\ 
\multirow[c]{2}{*}{$2$} & 3.409 & 3.389 & 3.455 & 3.560 \\ 
 & \footnotesize{(0.054)} & \footnotesize{(0.123)} & \footnotesize{(0.121)} & \footnotesize{(0.131)} \\ 
\multirow[c]{2}{*}{$3$} & 3.418 & 3.424 & 3.572 & 3.639 \\ 
 & \footnotesize{(0.054)} & \footnotesize{(0.125)} & \footnotesize{(0.224)} & \footnotesize{(0.240)} \\ 
\multirow[c]{2}{*}{$4$} & 3.424 & 3.466 & 3.605 & 3.757 \\ 
 & \footnotesize{(0.104)} & \footnotesize{(0.161)} & \footnotesize{(0.223)} & \footnotesize{(0.579)} \\ 
\multirow[c]{2}{*}{$5$} & 3.259 & 3.299 & 3.459 & 3.701 \\ 
 & \footnotesize{(0.059)} & \footnotesize{(0.067)} & \footnotesize{(0.245)} & \footnotesize{(0.871)} \\ 
 \bottomrule
\end{tabular}
\end{center}
\begin{center}
        \parbox{6in}{\footnotesize{Notes: The unit of measurement is seconds in Table~\ref{tab:computation-times}. Table~\ref{tab:computation-times} presents the average elapsed times for computing the penalized maximum statistic $T(\lambda, a) \equiv \max_{\gamma \in [-a,a]^4} [\widehat Q_n(\theta_{1,0},\gamma) - \lambda \|\gamma\|_1]$ across various $a$ and $\lambda$ values. The standard deviations are displayed in parentheses. The averages and standard deviations are computed based on 1000 repetitions.}}
    \end{center}
\[  \]
\[  \]
\[  \]
\[  \]
\end{table}
\newpage
\thispagestyle{empty}
\begin{table}[htbp]
   \centering
   \caption{Confidence Intervals Using Various Testing Methods}
   \label{tab:confidence-intervals}
   \medskip
   \begin{tabular}{C{0.14\textwidth} C{0.14\textwidth} C{0.14\textwidth} C{0.14\textwidth}}
    \toprule
    \multicolumn{2}{c}{Optim.} & \multicolumn{2}{c}{Unpen.}  \\
    \cmidrule(lr){1-2} \cmidrule(lr){3-4} 
    PSO & GS & PSO & GS  \\
    \midrule
    $[-0.30, 0.15]$ & $[-0.34, 0.15]$ & $[-0.43,0.13]$ & $[-0.58, 0.14]$   \\
    \bottomrule
   \end{tabular}
   \medskip
   \begin{center}
   \parbox{6.3in}{\footnotesize{Notes: The figures within the brackets denote the 95\% confidence intervals (CIs) derived from each testing and computation method. The first two columns labeled as 'Optim.' stands for our optimal CIs using the PSO and GS algorithms, respectively. Likewise, the next two columns represent the CIs from our unpenalized test.}}
    \end{center}
\end{table}

\newpage
\thispagestyle{empty}

\begin{table}[htbp]
\caption{Selected Optimal $\lambda$ Values across Different $\theta_1$ Values}
\label{tab:selection-optim-penalty}
\begin{center}
\begin{tabular}{C{0.12\textwidth} C{0.075\textwidth} C{0.075\textwidth} C{0.075\textwidth} C{0.075\textwidth} C{0.075\textwidth} C{0.075\textwidth} C{0.075\textwidth}}
\toprule
    \multirow[c]{2.5}{*}{stats.} & \multicolumn{7}{c}{$\theta_1 -\theta_{1,0}$} \\
     \cmidrule(lr){2-8}
     & $-$0.6 & $-$0.4 & $-$0.2 & 0 & 0.2 & 0.4 & 0.6 \\
    \toprule
    mean $\lambda$ & 0.281 & 0.266 & 0.237 & 0.208 & 0.253 & 0.353 & 0.412 \\ 
pos. prob. & 0.959 & 0.962 & 0.966 & 0.968 & 0.974 & 0.982 & 0.987 \\ 
    \bottomrule
    \end{tabular}
\end{center}
\begin{center}
    \parbox{5.6in}{Notes: The first row represents the averages of the estimated optimal penalty and the second row presents the probability of choosing a strictly positive penalty. Each column represents the optimal $\lambda$ choice at a given hypothesized value of $\theta_1$. 
    The number of simulation replications is 5000, and for each simulation draw, 5000 bootstrap replications are conducted.}
\end{center}
\end{table}

\newpage
\thispagestyle{empty}

\begin{table}[htbp]
\caption{Size and Power of Each Test under the Baseline Specification $(\bar{\pi}=2)$}
\label{tab:power-6-IVs-quadratic-spec}
\begin{center}
\begin{tabular}{C{0.225\textwidth} C{0.1\textwidth} C{0.1\textwidth} C{0.1\textwidth} C{0.1\textwidth} C{0.1\textwidth}}
    \toprule
    \multirow[c]{2.5}{*}{Tests} & \multicolumn{2}{c}{$\theta_{1,0}$} & \multicolumn{3}{c}{$\theta_{1,0}+ 2/\sqrt{n}$} \\ 
    \cmidrule(lr){2-3} \cmidrule(lr){4-6}
     & 4 IVs & 6 IVs & 4 IVs & 6 IVs & Diff. \\ 
    \midrule
    Optimal & 0.106 & 0.124 & 0.728 & 0.657 & 0.071 \\ 
    Unpenalized & 0.062 & 0.041 & 0.640 & 0.463 & 0.177 \\ 
    Unpen. w/ shift & 0.100 & 0.100 & 0.678 & 0.522 & 0.156 \\ 
    2SLS & 0.270 & 0.403 & 0.391 & 0.313 & 0.078 \\ 
    AR & 0.118 & 0.123 & 0.641 & 0.568 & 0.073 \\ 
    DL & 0.311 & 0.882 & 0.145 & 0.169 & $-$0.024 \\ 
\bottomrule
\end{tabular}
\end{center}

\begin{center}
    \parbox{5.5in}{\footnotesize{Notes: The nominal level is set at $0.1$. The first two columns denote the sizes of the tests using 4 and 6 IVs, respectively. The third and fourth columns show the powers of the tests under the local alternative $\theta_{1,0} + \frac{2}{\sqrt n} \approx \theta_{1,0} + 0.2$. The final column displays the reductions in power by including less informative IVs.
    The number of simulation draws is 5000, and for each simulation draw, 5000 bootstrap replications are conducted for the tests in the first three rows in the table.}}
\end{center}
\end{table}

\newpage
\thispagestyle{empty}

\begin{table}[htbp]
    \caption{Size and Power of Each Test under the Linear Specification $(\bar{\pi}=0)$}
\label{tab:power-6-IVs-linear-spec}
\begin{center}
\begin{tabular}{C{0.225\textwidth} C{0.1\textwidth} C{0.1\textwidth} C{0.1\textwidth} C{0.1\textwidth} C{0.1\textwidth}}
    \toprule
    \multirow[c]{2.5}{*}{Tests} & \multicolumn{2}{c}{$\theta_{1,0}$} & \multicolumn{3}{c}{$\theta_{1,0}+ 2/\sqrt{n}$} \\ 
    \cmidrule(lr){2-3} \cmidrule(lr){4-6}
     & 4 IVs & 6 IVs & 4 IVs & 6 IVs & Diff. \\ 
    \midrule
    Optimal & 0.120 & 0.129 & 0.617 & 0.572 & 0.045 \\ 
    Unpenalized & 0.054 & 0.039 & 0.436 & 0.320 & 0.116 \\ 
    Unpen. w/ shift & 0.100 & 0.100 & 0.428 & 0.381 & 0.047 \\ 
    2SLS & 0.269 & 0.380 & 0.438 & 0.371 & 0.067 \\ 
    AR & 0.113 & 0.108 & 0.684 & 0.633 & 0.051 \\ 
    DL & 0.425 & 0.908 & 0.113 & 0.222 & $-$0.109 \\ 
\bottomrule
\end{tabular}
\end{center}
\end{table}


\newpage
\clearpage
\appendix

\renewcommand{\thepage}{A-\arabic{page}}
\setcounter{page}{1}

\renewcommand{\theequation}{\thesection\arabic{equation}}
\setcounter{equation}{0}
\renewcommand{\thetable}{\thesection\arabic{table}}
\setcounter{table}{0}
\renewcommand{\thefigure}{\thesection\arabic{figure}}
\setcounter{figure}{0}

\section*{\Large Online Appendix to ``Inference for parameters identified by conditional
moment restrictions using a generalized Bierens
maximum statistic''}

\section{Testing Rational Unbiased Reporting of Ability Status}\label{sec:emp1}

\citet[BBCCR hereafter]{SSA:2004} examine whether
a self-reported disability status is a conditionally unbiased indicator of Social Security Administration (SSA)'s disability award decision.
Specifically, they test if $\tilde{U}_i = \tilde{A}_i - \tilde{D}_i$ has mean zero conditional on covariates $W_i$, where $\tilde{A}_i$ is the SSA disability award decision and $\tilde{D}_i$ is a self-reported disability status indicator.
Their null hypothesis is $H_0: \mathbb{E} [ \tilde{A}_i - \tilde{D}_i | W_i] = 0$, which is termed as the hypothesis of
\emph{rational unbiased reporting} of ability status (RUR hypothesis).\footnote{In this example, the null hypothesis is simpler than the empirical example in the main text because there is no parameter to estimate. There are alternative tests applicable \citep[e.g.,][among others]{escanciano2006consistent,shao2014martingale} but we have not tried to implement them.}
They use  a battery of tests, including a modified version of \citet{bierens1990consistent}'s original test, and conclude that they fail to reject the RUR hypothesis.
In fact, their Bierens test has the smallest $p$-value of 0.09 in their test results (see  Table II of their paper).
In this section, we revisit this result and apply our testing procedure.

Table \ref{Tab:twoway} shows a two-way table of $\tilde{A}_i$ and $\tilde{D}_i$
and Table \ref{Tab:sumstat} reports the summary statistics of $\tilde{A}_i$ and $\tilde{D}_i$ along those of covariates $W_i$.
After removing individuals with missing values in any of covariates, the sample size is $n=347$ and the number of covariates is $p=21$.\footnote{According to
 Table I in BBCCR, the sample size is 393 before removing observations with the missing values; however, there are only 388 observations in the data file archived at the Journal of Applied Econometrics web page.  After removing missing values, the size of the sample extract we use is $n=347$, whereas the originally reported sample size is $n=356$ in Table 2 in BBCCR.}

\begin{table}[htbp]
\caption{\label{Tab:twoway}Self-reported disability and SSA award decision}
\centering
\medskip
\begin{tabular}{lccr}  \hline \hline 
\vspace{-2ex}
& & &  \\ 
                                  & \multicolumn{2}{c}{Self-reported disability ($\tilde{D}$)} & Total  \\
SSA award decision ($\tilde{A}$)  &       0  &   1 & \\  \hline
0 &  35    & 51   &   86     \\
1 &  61   &  200  &   261  \\
Total  &  96  &  251 & 347  \\
\hline \end{tabular}
\end{table}

\begin{table}[htbp]
\caption{\label{Tab:sumstat}Summary Statistics}
\centering
\medskip
\begin{tabular}{lrrrrr}   \hline \hline
Variable & Mean & Stan. Dev. & Min. & Max. & $\widehat{\gamma}$ \\   \hline
SSA award decision ($\tilde{A}$) &	0.75  & 0.43  & 0  & 1  &  \\
Self-reported disability ($\tilde{D}$) &	0.72  & 0.45  & 0  & 1 &  \\   \hline
Covariates &  &  &  &  &  \\
White &	0.57  & 0.50  & 0  & 1  & 0.07 \\
Married &	0.58  & 0.49  & 0  & 1  & -0.16 \\
Prof./voc. training & 0.36  & 0.48  & 0  & 1  & 0.17 \\
Male &	0.39  & 0.49  & 0  & 1  & 0.02 \\
Age at application to SSDI & 55.97  & 4.81  & 33  & 76  & 0.33\\
Respondent income/1000 &	6.19  & 10.28  & 0  & 52  & 0.12 \\
Hospitalization &	0.88  & 1.44  & 0  & 14  & -- \\
Doctor visits &	13.12  & 13.19  & 0  & 90  & -- \\
Stroke &	0.07  & 0.26  & 0  & 1  & -0.90 \\
Psych. problems &	0.25  & 0.43  & 0  & 1  & -- \\
Arthritis &	0.40  & 0.49  & 0  & 1  & -- \\
Fracture &	0.13  & 0.33  & 0  & 1  & -- \\
Back problem &	0.59  & 0.49  & 0  & 1  & -0.13 \\
Problem with walking in room &	0.15  & 0.36  & 0  & 1 & --  \\
Problem sitting &	0.48  & 0.50  & 0  & 1  & -0.03 \\
Problem getting up &	0.59  & 0.49  & 0  & 1  & -0.03 \\
Problem getting out of bed &	0.24  & 0.43  & 0  & 1  & -0.13 \\
Problem getting up the stairs &	0.45  & 0.50  & 0  & 1  & -- \\
Problem eating or dressing &	0.07  & 0.26  & 0  & 1  & -- \\
Prop. worked in $t-1$ &	0.32  & 0.41  & 0  & 1  & 1.32 \\
Avg. hours/month worked &	2.68  & 8.85  & 0  & 60 & --  \\
\hline \end{tabular}
\end{table}

As in Section~\ref{sec:emp2}, we first studentize each of covariates and transform them by $x \mapsto \tan^{-1}(x)$ componentwise. 
The space $\Gamma$ is set as $\Gamma = [-10,10]^p$.
As mentioned before, to compute ${T}_{n}$ in \eqref{test-stat-def}, we use the \texttt{particleswarm} solver in \texttt{Matlab}.
It is computationally easier to obtain 
 ${T}_{n}$ in \eqref{test-stat-def}
  when $\lambda$ is relatively larger.
This is because a relevant space for $\Gamma$ is smaller with a larger $\lambda$.\footnote{The specification of $\Gamma = [-10,10]^p$ in this section is different from that of $\Gamma = [-5,5]^p$ in Section~\ref{sec:emp2}. Because $p$ is larger in the current example, we chose not to consider smaller values of $\lambda$ and also opted to use a larger $a$ in $\Gamma = [-a,a]^p$ to make sure that the value of $a$ is not binding in optimization.}

To choose an optimal $\lambda$ as described in Section \ref{sec:lambda:calibration},
we parametrize the null hypothesis $H_0: \mathbb{E} [ \tilde{U}_i | W_i] = 0$ by $g(X_i, \theta) = \tilde{U}_i - \theta$ in \eqref{eq:CM model} with $\theta_0 = 0$.
First note that $G_i = 1$ in this example. Therefore, for each $\lambda \in \Lambda$, $\mathcal{R}\left(\lambda,B\right)$ is an increasing function of $|B|$.
Thus, it suffices to evaluate the smallest value of $|B|$ satisfying $B \in \mathcal{B}$.
Here, we take it to the sample standard deviation of $\tilde{U}_i$.
For $\lambda$,
we take
$\Lambda = \{1, 0.9, \ldots, 0.2\}$. This range of $\lambda$'s is chosen by some preliminary analysis.
When $\lambda$ is smaller than 0.2, it is considerably harder to obtain stable solutions; thus, we do not consider smaller values of $\lambda$.
Since $\lambda \mapsto {T}_{n}(\lambda)$ is a decreasing function, we first start with the largest value of $\lambda$ and then solves sequentially by lowering the value of $\lambda$, while checking whether the newly obtained solution indeed is larger than the previous solution. This procedure results in a solution path by $\lambda$.

\begin{figure}[htbp]
	\caption{Testing Results}
	\label{fig-test-results}
	\begin{center}
		\includegraphics[scale=0.35]{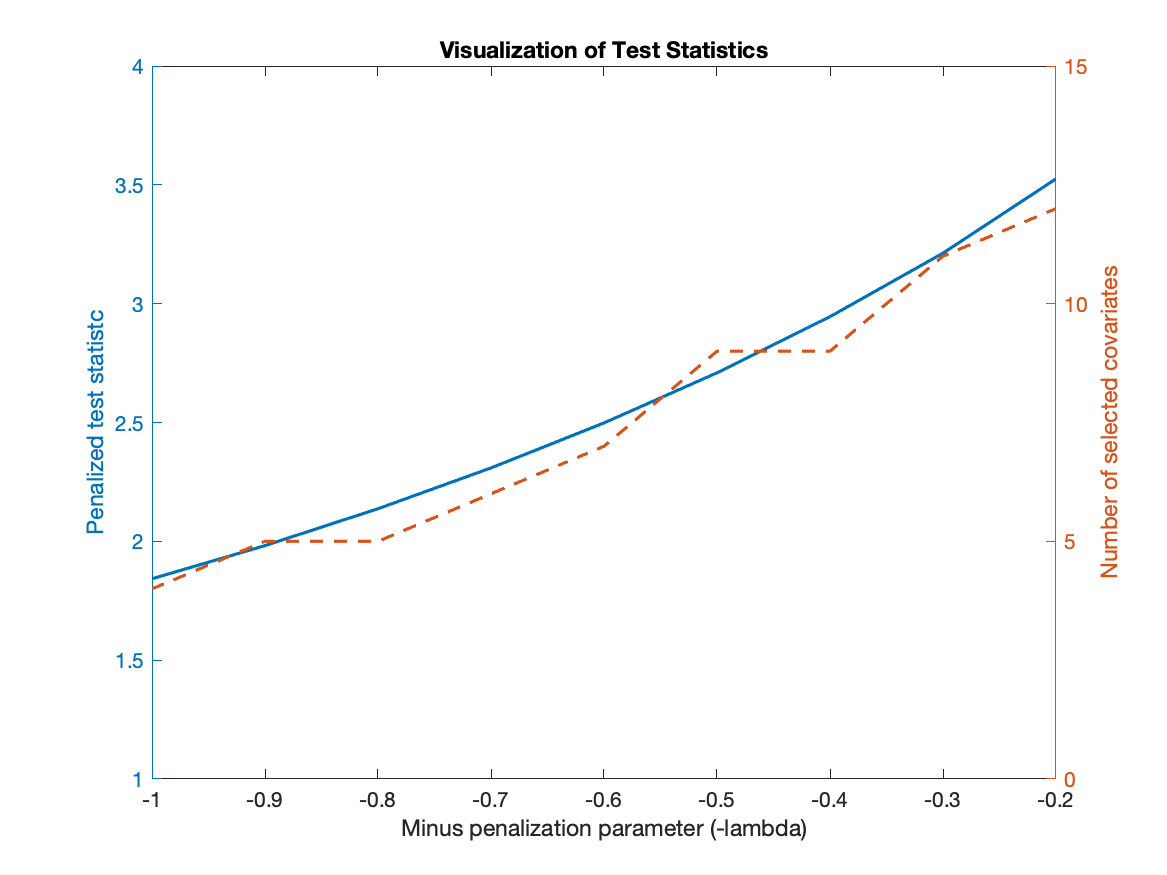} \hspace*{0.3ex}
		\includegraphics[scale=0.35]{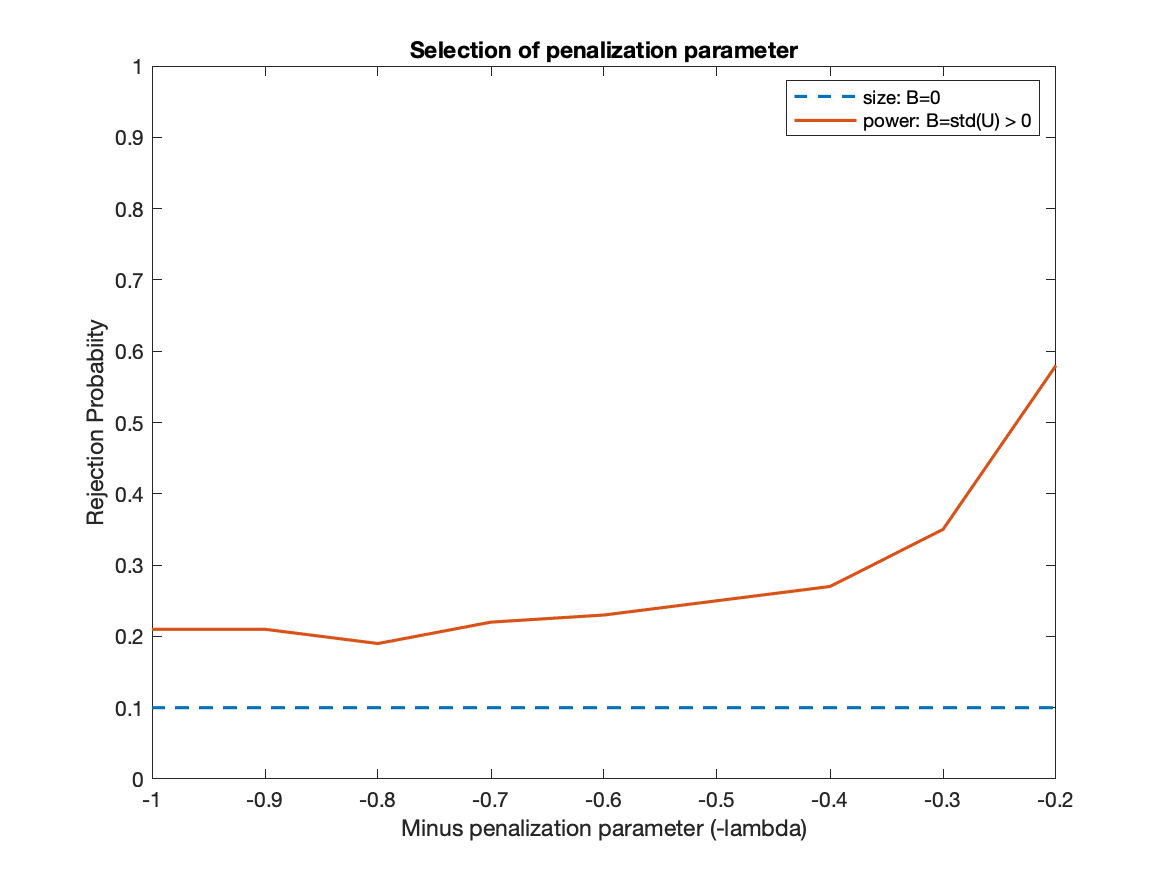} \\ \vspace*{0.5ex}
		\includegraphics[scale=0.72]{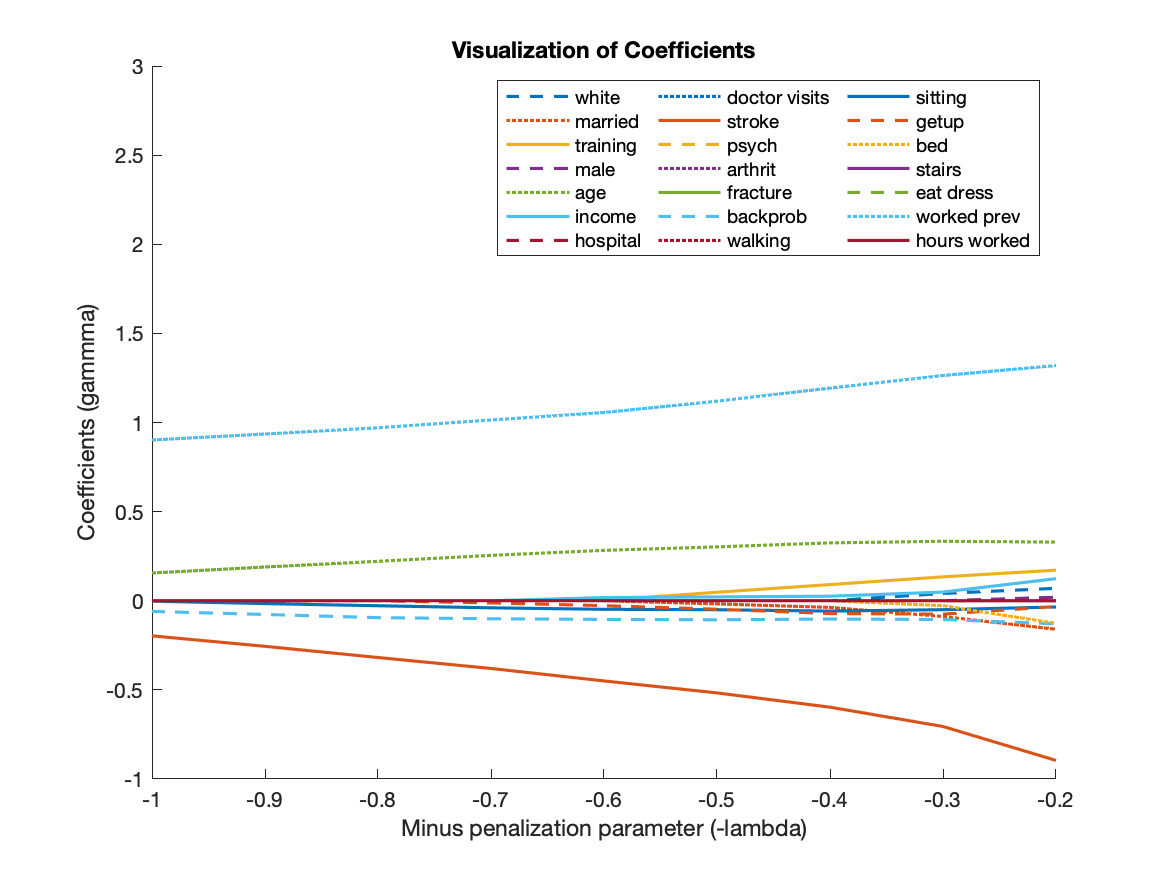}
	\end{center}
\end{figure}

Top-left panel of Figure \ref{fig-test-results} shows the solution path $\lambda \mapsto {T}_{n}(\lambda)$
along with the number of selected covariates, which is defined to be ones whose coefficients are no less than 0.01 in absolute value.  For the latter, 4 covariates are selected with $\lambda = 1$, whereas 12 are chosen with $\lambda = 0.2$.
Top-right panel displays the rejection probability defined in \eqref{RP:bootstrap} when
$B = 0$ (size) and $B = \widehat{\sigma}(\tilde{U}_i)$, where $\widehat{\sigma}(\tilde{U}_i)$ is the sample standard deviation of $\tilde{U}_i$.  The level of the test is 0.1 and there are 100 replications to compute the rejection probability.
The power is relatively flat up to $\lambda = 0.4$, increase a bit at $\lambda = 0.3$ and is maximized at
$\lambda = 0.2$.
The bottom panel visualizes each of 21 coefficients as $\lambda$ decreases.
It can be seen that the proportion of working in $t-1$  (\texttt{worked prev} in the legend of the figure)
has the largest coefficient (in absolute value) for all values of $\lambda$ and an indicator of stroke has the second largest coefficient, followed by age at application to Social Security Disability Insurance (SSDI).
The coefficients for selected covariates are given in the last column of Table \ref{Tab:sumstat} for $\lambda = 0.2$.

\begin{table}[htbp]
\caption{\label{Tab:bs-test}Bootstrap Inference}
\centering
\medskip
\begin{tabular}{cccc} \hline \hline
 $\lambda$   & Test statistic & No. of selected cov.s  & Bootstrap p-value \\    \hline
0.2 & 3.525 & 12 & 0.021     \\
0.3 & 3.213 & 11 & 0.020  \\

\hline \end{tabular}
\end{table}

Since the power in the top-right panel of Figure \ref{fig-test-results} is higher at $\lambda = 0.2$ and $0.3$, we report bootstrap test results for $\lambda \in \{0.2, 0.3\}$ in Table \ref{Tab:bs-test}.
There are $R_{\mathrm{PSO}}=1000$ bootstrap replications to obtain the bootstrap p-values.
Interestingly, we are able to reject  the RUR hypothesis at the 5 percent level, unlike BBCCR.
Furthermore, our analysis suggests that
the employment history, captured by the proportion of working previously,
stroke,
and the age at application to SSDI are the three most indicative covariates that point to the departure from the RUR hypothesis.

\end{document}